\newif\iflncs
\lncsfalse
\newif\ifsvjour
\svjourfalse

\iflncs
	\documentclass[envcountsame,envcountsect]{llncs}
\else \ifsvjour
	\documentclass[smallextended,numbook,envcountsame]{svjour3}
	\smartqed
\else
	\documentclass[11pt,a4paper]{article}
	\usepackage{fullpage}
\fi \fi
\usepackage[utf8]{inputenc}
\usepackage{enumerate}
\usepackage{amssymb}
\usepackage{amsmath}
\usepackage[breaklinks,bookmarks=false]{hyperref}
\usepackage[usenames,dvipsnames,svgnames,table]{xcolor}
\usepackage{bbm}

\iflncs
	\let\doendproof\endproof
	\renewcommand\endproof{~\hfill\qed\doendproof}
\else \ifsvjour
	\let\doendproof\endproof
	\renewcommand\endproof{~\hfill\qed\doendproof}
\else
	\usepackage[amsmath,amsthm,thmmarks,hyperref]{ntheorem}
\fi \fi
\usepackage{nicefrac}
\usepackage{framed}
\iflncs\else
	\usepackage{caption}
\fi

\usepackage{tikz}
\usetikzlibrary{plotmarks}
\usetikzlibrary{shadows}
\usetikzlibrary{decorations.pathreplacing}
\usetikzlibrary{positioning}
\usetikzlibrary{scopes}
\usetikzlibrary{arrows}
\usetikzlibrary{calc}
\usetikzlibrary{decorations.pathmorphing}

\tikzset{snake it/.style={decorate, decoration={snake,segment length=7pt, amplitude=1pt}}}

\tikzstyle{vertex}=[circle, draw,fill=gray!30, inner sep=0pt, minimum size=16pt]
\tikzstyle{svertex}=[circle, draw,fill=gray!30, inner sep=0pt, minimum size=10pt]
\tikzstyle{sgvertex}=[circle, draw,fill=gray!15, inner sep=0pt, minimum size=10pt]
\tikzstyle{ssgvertex}=[circle, draw,fill=gray!15, inner sep=0pt, minimum size=6pt]
\tikzstyle{cutb} = [fill=blue!30]
\tikzset{
>= stealth'}

\iflncs
	\expandafter\let\csname claim\endcsname\relax 
	\spnewtheorem{claim}{Claim}{\bfseries}{\rmfamily} 
	\spnewtheorem{fact}[theorem]{Fact}{\bfseries}{\itshape}
\else \ifsvjour
	\expandafter\let\csname claim\endcsname\relax 
	\newtheorem{claim}{Claim} 
	\newtheorem{fact}[theorem]{Fact}
\else
	\newtheorem{theorem}{Theorem}[section]
	
	\newtheorem{lemma}[theorem]{Lemma}
	\newtheorem{fact}[theorem]{Fact}

	\newtheorem{definition}[theorem]{Definition}
	\newtheorem{claim}{Claim} 

	\qedsymbol{\openbox}
\fi \fi
\newcommand{\qedmanual}{\hfill\ensuremath{\square}}

\DeclareMathOperator{\supp}{supp}

\DeclareMathOperator*{\con}{\mathcal{C}}
\DeclareMathOperator{\lb}{lb}
\DeclareMathOperator{\lbs}{\overline{\lb}}
\DeclareMathOperator{\OPT}{OPT}
\DeclareMathOperator{\LP}{LP}
\newcommand{\cO}{\mathcal{O}}

\newcommand{\eps}{\varepsilon}

\newcommand{\ceil}[1]{\lceil #1 \rceil}

\newcommand{\bR}{\mathbb{R}}
\newcommand{\one}{\mathbbm{1}}
\newcommand{\rb}[1]{\left( #1 \right)} 
\newcommand{\xs}{x^\star}

\newcommand{\Uisp}{U_i^\mathrm{sp}}
\newcommand{\gsp}{G_{\mathrm{sp}}}
\newcommand{\xsp}{\xs_{\mathrm{sp}}}
\newcommand{\xspp}{x'_{\mathrm{sp}}}
\newcommand{\wsp}{w_{\mathrm{sp}}}
\newcommand{\ysp}{y_{\mathrm{sp}}}
\newcommand{\yspp}{y_{\mathrm{sp}}'}
\newcommand{\giaux}{G_i^{\mathrm{aux}}}
\newcommand{\tG}{\widetilde{G}}
\newcommand{\Bs}{B^{\star}}
\newcommand{\alphas}{\alpha^{\star}}

\DeclareMathOperator{\debt}{debt}
\DeclareMathOperator{\bad}{bad}
\DeclareMathOperator{\tail}{tail}
\DeclareMathOperator{\head}{head}

\begin{document}

\iflncs
\else \ifsvjour
	\title{Constant Factor Approximation for ATSP with Two Edge Weights}
	\author{Ola Svensson \and Jakub Tarnawski \and L\'aszl\'o A. V\'egh}
	\institute{\email{ola.svensson@epfl.ch} \and \'Ecole Polytechnique F\'ed\'erale de Lausanne. \and {Supported by ERC Starting Grant 335288-OptApprox.} \\ \email{jakub.tarnawski@epfl.ch} \and \'Ecole Polytechnique F\'ed\'erale de Lausanne. \and {Supported by ERC Starting Grant 335288-OptApprox.} \\ \email{L.Vegh@lse.ac.uk}\and London School of Economics.\and {Supported by EPSRC First Grant EP/M02797X/1.}}
\else
	\title{Constant Factor Approximation\\ for ATSP with Two Edge Weights}
	\author{Ola Svensson\thanks{\'Ecole Polytechnique F\'ed\'erale de Lausanne. Supported by ERC Starting Grant 335288-OptApprox.}\\
	\texttt{ola.svensson@epfl.ch}
	\and
	Jakub Tarnawski\thanks{\'Ecole Polytechnique F\'ed\'erale de Lausanne.}\\
	\texttt{jakub.tarnawski@epfl.ch}
	\and
	L\'aszl\'o A. V\'egh\thanks{London School of Economics.  Supported by EPSRC First Grant EP/M02797X/1.} \\
	\texttt{l.vegh@lse.ac.uk}
	}
\fi \fi

\date{}
 
\maketitle

\begin{abstract}
We give a constant factor approximation algorithm for the Asymmetric Traveling Salesman Problem on shortest path metrics of directed graphs with two different edge weights. For the case of unit edge weights, the first constant factor approximation was given recently by Svensson. This was accomplished by introducing an easier problem called Local-Connectivity ATSP and showing that a good solution to this problem can be used to obtain a constant factor approximation for ATSP.
In this paper, we solve Local-Connectivity ATSP for two different edge weights. The solution is based on a flow decomposition theorem for solutions of the Held-Karp relaxation, which may be of independent interest.
\end{abstract}

\section{Introduction}
\label{sec:intro}
 
The traveling salesman problem --- one of finding the shortest tour of
$n$ cities --- is one of the most classical optimization problems. Its
definition  dates back to the 19th century and since then a large body of work has been
devoted to designing ``good'' algorithms using heuristics, mathematical
programming techniques, and approximation algorithms. The focus of this work is
on approximation algorithms. A natural and necessary assumption in this line of
work that we also make throughout this paper  is that the distances satisfy the
triangle inequality: for any triple $i,j,k$ of cities, we have $d(i,j) + d(j,k)
\geq d(i,k)$ where $d(\cdot, \cdot)$ denotes the pairwise distances between
cities. In other words, it is not more expensive to take the direct path
compared to a path that makes a detour.


With this assumption, the approximability of TSP turns out to be a very delicate
question that has attracted significant research efforts. Specifically, 
one of the first approximation algorithms (Christofides' heuristic~\cite{Christofides76}) was designed for
the \emph{symmetric} traveling salesman problem (STSP) where we assume
symmetric distances ($d(i, j) = d(j, i)$). Several works
(see e.g.~\cite{FriezeGM82,AsadpourGMGS10,Oveis11,Anari14,Svensson15}) have addressed the more
general \emph{asymmetric} traveling salesman problem (ATSP) where we make no
such assumption. 
%
%

However, there are still large gaps in our understanding of  both
STSP and ATSP.  In fact, for STSP, the best approximation algorithm remains
Christofides' $3/2$-approximation algorithm  from the
70's~\cite{Christofides76}. For the harder ATSP, the state of the art is
a $\cO(\log n/ \log \log n)$-approximation algorithm by Asadpour et al.~\cite{AsadpourGMGS10} and
a recent  $\cO(\mbox{poly} \log \log n)$-estimation
algorithm\footnote{An estimation
algorithm is a polynomial-time algorithm for approximating/estimating the
optimal value without necessarily finding a solution to the problem.}
by Anari and Oveis Gharan~\cite{Anari14}.  On the
negative side, the best inapproximability results only  say that STSP and ATSP
are hard to approximate within  factors ${123}/{122}$ and ${75}/{74}$,
respectively~\cite{KarpinskiLS15}. Closing these gaps is a major open problem
in the field of approximation algorithms (see e.g.  ``Problem 1'' and ``Problem
2'' in the list of open problems in the recent book by Williamson and
Shmoys~\cite{WSbook}).
What is perhaps even more intriguing about these questions is that we expect that
a standard linear programming (LP) relaxation,
%
%
%
often referred to as the Held-Karp relaxation, already gives better guarantees.
Indeed, it is conjectured to give a guarantee of $4/3$ for STSP and
a guarantee of $\cO(1)$ (or even $2$) for ATSP.  

An equivalent formulation of STSP and ATSP  from a more graph-theoretic point of view is the following.
For STSP, we are given a weighted undirected graph $G=(V,E,w)$ where $w:
E \rightarrow \mathbb{R}_{+}$ and we wish to find a multisubset $F$ of
edges of minimum total weight such that  $(V,F)$ is connected and Eulerian.
Recall that an undirected graph is Eulerian if every vertex has even degree. We also
remark that we use the term multisubset as the solution $F$ may use the same
edge several times. An intuitive point of view on this definition is that $G$
represents a road network, and a solution is a tour that visits each vertex at
least once (and may use a single edge/road several times). The definition of ATSP
is similar, with the differences that the input graph is directed and the
output is Eulerian in the directed sense: the in-degree of each vertex equals
its out-degree.
Having defined the traveling salesman problem in this way, there
are several natural special cases to consider. For example,  what if $G$ is
planar? Or, what if all the edges/roads have the same length, i.e., if $G$ is unweighted?

For planar graphs,  we have much better algorithms than in general. Grigni,
Koutsoupias and Papadimitriou~\cite{GrigniKP95} first obtained a polynomial-time approximation
scheme for STSP restricted to unweighted planar graphs, which was later
generalized to edge-weighted planar graphs by Arora et al.~\cite{AroraGKKW98}.
More recently, ATSP on planar graphs (and more generally bounded genus graphs)
was shown to admit constant factor approximation algorithms (first by Oveis
Gharan and Saberi~\cite{Oveis11} and later by Erickson and Sidiropoulos~\cite{EricksonS14} who improved the dependency on the
genus).

In contrast to planar graphs, STSP and ATSP remain APX-hard for unweighted
graphs (ones where all edges have identical weight) and, until recently, there were no better algorithms for these cases.
Then, in a recent series of papers, the approximation guarantee of $3/2$ was
finally improved for STSP restricted to unweighted graphs. Specifically,
Oveis Gharan, Saberi and Singh~\cite{GharanSS11} first gave an approximation
guarantee of $1.5-\epsilon$; M\"omke and Svensson~\cite{MomkeS11} proposed
a different approach yielding a $1.461$-approximation guarantee;
Mucha~\cite{Mucha12} gave a tighter analysis of this algorithm; and Seb\H {o}
and Vygen~\cite{SeboV14} significantly developed the approach to give the
currently best approximation guarantee of $1.4$.  
Similarly, for ATSP, it was only very recently that the
restriction to unweighted graphs could be leveraged: the first constant
approximation guarantee for unweighted graphs was given by Svensson~\cite{Svensson15}. 
In this paper we make progress towards the general problem by taking the logical next step and addressing a simple case left unresolved by  \cite{Svensson15}: graphs with two different edge weights.

\begin{theorem}
  There is an $\cO(1)$-approximation algorithm for ATSP on graphs with two
  different edge weights.
  \label{thm:mainintro}
\end{theorem}

The paper \cite{Svensson15} introduces an ``easier'' problem named Local-Connectivity ATSP, where one needs to find an Eulerian multiset of edges crossing only sets in a given partition rather than all possible sets (see next section for definitions). It is shown that an ``$\alpha$-light" algorithm to this problem yields a $(9+\varepsilon)\alpha$-factor approximation for ATSP. For unweighted graphs (and slightly more generally, for node-induced weight functions\footnote{For ATSP, we can think of a node-weighted graph as an
edge-weighted graph where the weight of an edge $(u,v)$ equals the node weight
of $u$.}) it is fairly easy to obtain a 3-light algorithm for Local-Connectivity ATSP; the difficult part in  \cite{Svensson15} is the black-box reduction of ATSP to this problem. Note that \cite{Svensson15} easily gives an $\cO(w_{\max}/w_{\min})$-approximation algorithm in general if we take
$w_{\max}$ and $w_{\min}$ to denote the largest and smallest edge weight, respectively. However, obtaining a constant factor approximation even for two different weights requires substantial further work. 

In Local-Connectivity ATSP we need a lower bound function $\lb:V\to \mathbb{R}_+$ on the vertices. The natural choice for node-induced weights is $\lb(v)=\sum_{e\in \delta^+(v)} w(e)x^*_e$. With this weight function, every vertex is able to ``pay'' for the incident edges in the Eulerian subgraph we are looking for.  This choice of $\lb$ does not seem to work for more general weight functions, and we need to define $\lb$ more ``globally'', using a new flow theorem for Eulerian graphs (Theorem\nobreakspace \ref {thm:flow_thm}).
In Section\nobreakspace \ref {sec:overview}, after the preliminaries, we give a more detailed
overview of these techniques and the proof of the theorem.
 Our argument is somewhat technical, but it demonstrates the potential of the Local-Connectivity ATSP problem as a tool for attacking general ATSP.

Finally, let us remark that both STSP~\cite{PapYan93,BermanK06} and
ATSP~\cite{Blaser04}  have been studied in the case  when all distances are
either $1$ or $2$.  That restriction is very different from our setting, as in
those cases the input graph is complete. In particular,  it is 
trivial to get  a $2$-approximation algorithm there, whereas in our
setting -- where the input graph is \emph{not} complete -- a constant
factor approximation guarantee already requires non-trivial algorithms.
(In our setting, we can still think about the metric completion,
but it will usually have more than two different edge weights.)

\subsection{Notation and preliminaries}
\label{sec:prelim}

We consider an edge-weighted directed graph $G=(V,E,w)$ with $w:E\to \mathbb{R}_+$. 
For a vertex subset $S\subseteq V$ we let $\delta^+(S)=\{(u,v)\in E: u\in S, v\in V\setminus S\}$ and
$\delta^-(S)=\{(u,v)\in E: u\in V\setminus S, v\in S\}$ denote the
sets of outgoing and incoming edges, respectively. For two vertex
subsets $X,Y\subseteq V$, we let $\delta(X,Y)=\{(u,v)\in E: u\in
X\setminus Y, v\in Y\setminus X\}$. For a subset of edges $E'\subseteq E$, we use $\delta^+_{E'}(S)=\delta^+(S)\cap E'$ and  $\delta^-_{E'}(S)=\delta^-(S)\cap E'$.
We also let $\con(E')=(\tilde G_1,\ldots,\tilde G_k)$ denote the set of weakly connected components of the graph $(V,E')$; the vertex set $V$ will always be clear from the context. For a directed graph $\tilde G$ we use $V(\tilde G)$ to denote its vertex set and $E(\tilde G)$ the edge set.
For brevity, we denote the
singleton set $\{v\}$ by $v$ (e.g. $\delta^+(v)=\delta^+(\{v\}))$, and we use the notation $x(F) = \sum_{e\in F} x_e$
for a subset $F \subseteq E$ of edges
and a vector $x \in \mathbb{R}^E$.
For a multiset $F$, we have $\one_F$ denote the indicator vector of $F$,
which has a coordinate for each edge $e$ with value equal to the number of copies of 
$e$ in $F$.
For the case of two edge weights, we use $0 \le w_0 < w_1$ to denote the two possible values, and partition $E=E_0\cup E_1$ so that $w(e)=w_0$ if $e\in E_0$ and $w(e)=w_1$ if $e\in E_1$. We will refer to edges in $E_0$ and $E_1$  as cheap and expensive edges, respectively.

We define
ATSP as  the problem of finding a connected Eulerian subgraph of minimum weight. 
As already mentioned in the introduction, this definition is equivalent to that
of visiting each city exactly once (in the metric completion) since we assume
the triangle inequality. The formal definition is as follows.
\begin{framed}
\begin{center}\textbf{ATSP} \end{center}
\begin{description}
  \item[\textnormal{\emph{Given:}}]An edge-weighted (strongly connected) digraph $G=(V,E, w)$.  
	
\item[\textnormal{\emph{Find:}}] A 
  multisubset $F$ of $E$ of minimum total weight $w(F) = \sum_{e\in F} w(e)$ such that $(V,F)$ is Eulerian and connected.
  \end{description}
\end{framed}

\paragraph{Held-Karp Relaxation.}
The Held-Karp relaxation has a variable $x_e \geq 0$ for every edge in $G$. The intended meaning is that $x_e$ should equal the number of times $e$ is used in the solution. The relaxation $\LP(G)$ is defined as follows:

\begin{equation}
\begin{aligned}
\arraycolsep=1.4pt\def\arraystretch{1.2}
\begin{array}{lrlr}
\mbox{minimize} \qquad & \displaystyle \sum_{e\in E} w(e) x_e  \\[7mm] 
\mbox{subject to} \qquad  & \displaystyle x(\delta^+(v)) = & \displaystyle x(\delta^-(v)) & v\in V, \\  
& \displaystyle x(\delta^+(S)) \geq & 1 & \emptyset \neq S \subsetneq V, \\
& x \geq & 0.
\end{array}
\end{aligned}\tag{$\LP(G)$}
\end{equation}
The first set of constraints says that the in-degree should equal the
out-degree for each vertex, i.e., the solution should be Eulerian. The second
set of constraints enforces that the solution is connected; they are
sometimes referred to as subtour elimination constraints.  Finally, we remark
that although the Held-Karp relaxation has exponentially many constraints, it
is well-known that we can solve it in polynomial time either by using the
ellipsoid method with a separation oracle or by formulating an equivalent
compact (polynomial-size) linear program.
We will use $x^*$ to denote an optimal solution to $\LP(G)$ of value $\OPT$, which is a lower bound on the value of an optimal solution to ATSP on $G$.

\paragraph{Local-Connectivity ATSP.}
The Local-Connectivity ATSP problem can be seen as a two-stage procedure. In the first stage, the input is an edge-weighted digraph $G=(V,E, w)$ and the output is a ``lower bound'' function  $\lb: V \rightarrow \mathbb{R}_+$  on the vertices such that $\lb(V)\le \OPT$. In the second stage, the input is a partition of the vertices, and the output is an Eulerian multisubset of edges which crosses each set in the partition and where the ratio of weight to $\lb$ of every connected component is as small as possible. 
We now give the formal description of the second stage, assuming the  $\lb$ function is already computed.
\begin{framed}
\begin{center} \textbf{Local-Connectivity ATSP} \end{center}
\begin{description}
  \item[\textnormal{\emph{Given:}}]An edge-weighted 
	digraph $G=(V,E, w)$, a function $\lb: V \rightarrow \mathbb{R}_+$ with $\lb(V)\le \OPT$, and   a partitioning $V = V_1 \cup V_2 \cup \ldots \cup V_k$
	of the vertices.
	
\item[\textnormal{\emph{Find:}}]A Eulerian multisubset $F$ of $E$   such
that
\begin{align*}
  |\delta^+_{F}(V_i)| \geq 1 \ \ \mbox{ for } i=1, 2, \ldots, k \qquad \mbox{and}
  \qquad \max_{\tilde G\in \con(F)} \frac{w(\tilde G)}{\lb(\tilde G)} \mbox{ is minimized.}
\end{align*}
\end{description}
\end{framed}
Here we used the notation that for a connected component $\tilde G$ of $(V,F)$, $w(\tilde G)=\sum_{e\in E(\tilde G)} w(e)$ (summation over the edges) and $\lb(\tilde G)=\sum_{v\in V(\tilde G)} \lb(v)$ (summation over the vertices).
We say that an algorithm for Local-Connectivity ATSP   is 
\emph{$\alpha$-light} on $G$ if it is guaranteed, for any partition, to find a solution $F$
such that  for every component $\tilde G\in \con(F)$,
\iflncs
\else
\begin{align*}
  \frac{w(\tilde G)}{\lb(\tilde G)} \leq \alpha.
\end{align*}
\fi

In \cite{Svensson15}, $\lb$ is defined as $\lb(v)=\sum_{e\in \delta^+(v)} w(e)x^*_e$; note that $\lb(V)=OPT$ in this case.
%
We remark that we use the ``$\alpha$-light'' terminology to avoid any ambiguities
with the concept of approximation algorithms (an $\alpha$-light
algorithm does not compare its solution to an optimal solution
to the given instance of Local-Connectivity ATSP). 

Perhaps the main difficulty of ATSP is to satisfy the connectivity requirement,
i.e., to select an Eulerian subset $F$ of edges which connects the whole graph.
{Local-Connectivity} ATSP relaxes this condition -- we only need to find an
Eulerian set $F$ that crosses the $k$ cuts defined by the partition.  This
makes it intuitively an ``easier'' problem than ATSP. Indeed, an
$\alpha$-approximation algorithm for ATSP (with respect to the Held-Karp
relaxation) is trivially an $\alpha$-light algorithm for Local-Connectivity
ATSP for an arbitrary $\lb$ function with $\lb(V) = OPT$: just return the same Eulerian subset $F$ as the algorithm for ATSP; since the
set $F$ connects the graph, we have $\max_{\tilde G \in \con(F)} w(\tilde
G)/\lb(\tilde G) = w(F)/\lb(V) \leq \alpha$.  Perhaps more surprisingly, the
main technical theorem of~\cite{Svensson15} shows that the two problems are
equivalent up to small constant factors.  

\begin{theorem}[\cite{Svensson15}]
Let  $\mathcal{A}$ be an algorithm for Local-Connectivity ATSP.
  Consider an ATSP instance $G=(V,E,w)$, and let $\OPT$ denote the optimum value of the Held-Karp relaxation. 
  If $\mathcal{A}$ is $\alpha$-light on $G$, then there exists a tour of $G$ with value at
  most $5\alpha\OPT$. 
%
  Moreover, for any $\varepsilon >0$,  a tour of value at most
  $(9+\varepsilon)\alpha\OPT$ can be found in time polynomial in the number
  $n=|V|$ of vertices, in $1/\varepsilon$,  and in the running time of $\mathcal{A}$.
  \label{thm:LoctoGlo}
\end{theorem}
In other words, the above theorem says that in order to approximate an  ATSP
instance $G$, it is sufficient to devise a polynomial-time algorithm to
calculate  a lower bound $\lb$ and a polynomial time algorithm
for Local-Connectivity ATSP that is $\cO(1)$-light on $G$ with respect to this $\lb$ function.
Our main result is proved using this framework.

\subsection{Technical overview}
\label{sec:overview}

\paragraph{Singleton partition.} Let us start by outlining the fundamental ideas of our algorithm and comparing it to \cite{Svensson15} for the special case of Local-Connectivity ATSP when all partition classes $V_i$ are singletons.
For unit weights, the choice $\lb(v)=\sum_{e\in \delta^+(v)} w(e)\xs_e=\xs(\delta^+(v))$ in \cite{Svensson15} is a natural one: intuitively, every node is able to pay for its outgoing edges. We can thus immediately give an algorithm for this case: just select an arbitrary integral solution $z$ to the circulation problem with node capacities $1 \le z(\delta^+(v))\le \ceil{\xs(\delta^+(v))}$. Then for any $v$ we have $z(\delta^+(v))\le \xs(\delta^+(v))+1\le 2\xs(\delta^+(v))$ and hence $\sum_{e\in \delta^+(v)} w(e)z_e\le 2\lb(v)$, showing that $z$ is a 2-light solution.

The same choice of $\lb$ does not seem to work in the presence of two different edge costs. Consider a case when every expensive edge carries only a small fractional amount of flow. Then $\sum_{e\in \delta^+(v)} w(e)\xs_e$ can be much smaller than the expensive edge cost $w_1$, and thus the vertex $v$ would not be able to ``afford'' even a single outgoing expensive edge. To resolve this problem, we bundle small fractional amounts of expensive flow, channelling them to reach a small set of terminals. This is achieved via Theorem\nobreakspace \ref {thm:flow_thm}, a flow result which might be of independent interest. It shows that within the fractional Held-Karp solution $\xs$, we can send the flow from an arbitrary edge set $E'$ to a sink set $T$ with $|T|\le 8\xs(E')$; in fact, $T$ can be any set minimal for inclusion such that it can receive the total flow from $E'$.
We apply this theorem for $E'=E_1$, the set of expensive edges; let $f$ be the flow from $E_1$ to $T$, and call elements of $T$ \emph{terminals}. Now, whenever an expensive edge is used, we will ``force'' it to follow $f$ to a terminal in $T$, where it can be paid for.
Enforcement is technically done by splitting the vertices into two copies, one carrying the $f$ flow and the other the rest.
Thus we obtain the \emph{split graph} $\gsp$ and split fractional optimal solution $\xsp$. 

The design of the split graph is such that every walk in it which starts with an expensive edge must proceed through cheap edges until it reaches a terminal before visiting another expensive edge. In our terminology, expensive edges create ``debt'', which must be paid off at a terminal. 
Starting from an expensive edge, the debt must be carried until a terminal is reached, and no further debt can be taken in the meantime.
The bound on the number of terminals guarantees that we can assign a lower bound function $\lb$ with $\lb(V)\le \OPT$ such that (up to a constant factor) cheap edges are paid for locally, at their heads, whereas expensive edges are paid for at the terminals they are routed to.
Such a splitting easily solves Local-Connectivity ATSP for the singleton partition: find an arbitrary integral circulation $z_{\mathrm{sp}}$ in the split graph with an upper bound
 $z_{\mathrm{sp}}(\delta^+(v))\le \lceil 2\xsp(\delta^+(v))\rceil$ on every node, and a lower bound $1$ on whichever copy of $v$ transmits more flow. Note that $2\xsp$ is a feasible fractional solution to this problem. We map $z_{\mathrm{sp}}$ back to an integral circulation $z$ in the original graph by merging the split nodes, thus obtaining a constant-light solution.

\paragraph{Arbitrary partitions.} 
Let us now turn to the general case of Local-Connectivity ATSP, where the input is an arbitrary partition $V=V_1\cup\ldots\cup V_k$. 
For unit weights this is solved in~\cite{Svensson15} via an integer circulation problem on a modified graph. Namely, an auxiliary node $A_i$ is added to represent each partition class $V_i$, and one unit of in- and outgoing flow from $V_i$ is rerouted through $A_i$. In the circulation problem, we require exactly one in- and one outgoing edge incident to $A_i$ to be selected.
When we map the solution back to the original graph, there will be one incoming and one outgoing arc from every set $V_i$ (thus satisfying the connectivity requirement) whose endpoints inside $V_i$ violate the Eulerian condition.
In \cite{Svensson15} every $V_i$ is assumed to be strongly connected, and therefore we can  ``patch up'' the circulation by connecting the loose endpoints by an arbitrary path inside $V_i$. This argument easily gives a 3-light solution.

Let us observe that the strong connectivity assumption is in fact not needed for the result in \cite{Svensson15}. Indeed, given a component $V_i$ which is not strongly connected, consider its decomposition into strongly connected (sub)components, and pick a $U_i \subseteq V_i$ which is a sink (i.e. it has no edges outgoing to $V_i \setminus U_i$). We proceed by rerouting $1$ unit of flow through a new auxiliary vertex just as in that algorithm, but we do this for $U_i$ instead. This guarantees that $U_i$ has at least one outgoing edge in our solution, and that edge must leave $V_i$ as well. 

Our result for two different edge weights takes this observation as the starting point, but the argument is much more complicated. We will find an integer circulation in a graph based on the split graph $\gsp$, and for every $1\le i\le k$, there will be an auxiliary vertex $A_i$ representing a certain subset $U_i\subseteq V_i$. These sets $U_i$ will be obtained as sink components in certain auxiliary graphs we construct inside each $V_i$. This construction is presented in Section\nobreakspace \ref {sec:solvlcATSP}; we provide a roadmap to the construction at the beginning of that section.

\iflncs
\else
\section{The Flow Theorem}
\label{sec:flow}
In this section we prove our main flow decomposition result. As indicated in Section\nobreakspace \ref {sec:overview}, we will use it to channel the flow from the expensive edges $E_1$ to a small set of terminals $T$ (where $|T|\le 8w(E_1)$). We will use the theorem stated below by moving the tail of every edge in $E_1$ to a new vertex $s$. If $w(E_1)\ge 1$, then the constraints of the Held-Karp relaxation guarantee condition \eqref{cond:degree}.
The details of the reduction are given in Lemma\nobreakspace \ref {lem:flow_app}.

\begin{theorem} \label{thm:flow_thm}
Let $D = (V\cup\{s\},E)$ be a directed graph, let $c : E \to \bR_+$ be a nonnegative capacity vector, and let $s$ be a source node with no incoming edges, i.e., $\delta^-(s) = \emptyset$.
Assume that for all $\emptyset \neq S\subseteq V$ we have
\begin{equation}
c(\delta^-(S))\ge \max\{1,c(\delta^+(S))\}. \label{cond:degree}
\end{equation}
Consider a set $T \subseteq V$ such that there exists a flow $f \le c$ of value $c(\delta^+(s))$ from the source $s$ to the sink set $T$, and $T$ is minimal subject to this property.\footnote{That is, the maximum flow value from $s$ to any proper subset $T'\subsetneq T$ is smaller than $c(\delta^+(s))$.} Then $|T| \le 8c(\delta^+(s))$.
\end{theorem}

The proof of this theorem can be skipped on first reading, as the algorithm in Section\nobreakspace \ref {sec:lcATSP} only uses it in a black-box manner.
\begin{proof}
Fix a minimal set $T$ and denote $k = |T|$. Our goal is to prove that $k \le 8 c(\delta^+(s))$. We know that there exists a flow of value $c(\delta^+(s))$ from $s$ to $T$. For any such flow $f$ we define its \emph{imbalance sequence} to be the sequence of values $z(t) = f(\delta^-(t)) - f(\delta^+(t)) \in \bR_+$ for all $t \in T$ sorted in non-increasing order. We select the flow $f$ which maximizes the imbalance sequence (lexicographically). We write $T = \{t_1, \ldots, t_k\}$ so that $z(t_1) \ge z(t_2) \ge \ldots \ge z(t_k)$; denote $z_i = z(t_i)$ for brevity. By minimality of $T$ we have $z(t_k) > 0$.
The following is our main technical lemma.
\begin{lemma} \label{lem:not_too_many_small}
Let $\ell$ be the number of $t \in T$ with $z(t) \ge \frac 14$, i.e., $z_1 \ge ... \ge z_\ell \ge \frac 14 > z_{\ell + 1} \ge ... \ge z_k$. Then we have \[ \frac 14 (k-\ell) \le \sum_{i=1}^\ell z_i. \]
In other words, the number of terminals with small imbalance is not much more than the sum of large imbalances.
\end{lemma}

Assuming this lemma, the main theorem follows immediately, since
we have $\frac 14 k = \frac 14 \ell + \frac 14 (k - \ell) \le \sum_{i=1}^\ell z_i + \sum_{i=1}^\ell z_i \le 2c(\delta^+(s))$, i.e., $k \le 8c(\delta^+(s))$.
\end{proof}

The remainder of this section is devoted to the proof of the technical lemma.
Let us first give an outline.  We analyze the residual capacity (with respect to $f$) of certain cuts that must be present due to the lexicographic property. First of all, there must be a saturated cut (that is, one  of $0$ residual  in-capacity) $A\subseteq V$ containing all large terminals (i.e., those with imbalance at least $1/4$) but no small ones (\protect \MakeUppercase {C}laim\nobreakspace \ref {cl:A}).
Next, consider an arbitrary \emph{small} terminal $i$. Also by the maximality property, it is not possible to increase the value of $z_i$ to $1/4$ by rerouting flow from other small terminals to $i$.
Hence there must be a cut $B_i$, disjoint from $A$, which contains  $t_i$ as the only terminal and has residual in-capacity less than $1/4-z_i$ in  $D\setminus A$  (\protect \MakeUppercase {C}laim\nobreakspace \ref {cl:Bi}). As an illustration of the argument, let us assume that these sets $B_i$ are pairwise disjoint. It follows from \eqref{cond:degree} that the residual in-capacity of $B_i$ is at least $1-z_i$ (\protect \MakeUppercase {C}laim\nobreakspace \ref {cl:residual-one-term-set}). 
Hence every set $B_i$ must receive $3/4$ units of  its residual in-capacity from $A$. On the other hand, \eqref{cond:degree} upper-bounds the residual out-capacity of $A$ by $2\sum_{i=1}^\ell z_i$ (\protect \MakeUppercase {C}laim\nobreakspace \ref {cl:residual-term-set}). These together give a bound $\frac{3}{4}(k-\ell)\le 2 \sum_{i=1}^\ell z_i$.  Recall however that we assumed that all sets $B_i$ are disjoint. Since these sets may in fact overlap, the proof needs to be more careful: instead of sets $B_i$, we argue with the sets $B_i\setminus (\cup_{j\neq i} B_j)$ (nonempty as containing  $t_i$), and the union of pairwise intersections $B^*$; thus instead of $3/8$, we get a slightly worse constant $1/4$.

\ifsvjour
	\begin{proof}[of Lemma\nobreakspace \ref {lem:not_too_many_small}]
\else
	\begin{proof}[Proof of Lemma\nobreakspace \ref {lem:not_too_many_small}]
\fi
First note that the claim is trivial if $\ell = k$, so assume $\ell < k$.
For an arc $e=(u,v)$, we let $\overleftarrow{e}=(v,u)$ denote the reverse arc. We define the residual graph $D_f=(V\cup\{s\},E_f)$ with $E_f=\{e\in E: f(e)<c(e)\}\cup\{e: \overleftarrow{e}\in E, f(\overleftarrow e)>0\}$. The residual capacity for the first set of arcs is defined as $c_f(e)=c(e)-f(e)$, and for the second set as $c_f(e)=f(\overleftarrow{e})$.
For any set $X \subseteq V$, and a disjoint $Y \subseteq V$, let $\delta_f^-(X)$, $\delta_f^+(X)$ and $\delta_f(X,Y)$ denote the capacities of the respective cuts in the residual graph $D_f$ of $f$, i.e.,
\begin{align*}
\delta_f^-(X) &= \sum_{e\in \delta^-(X)} c_f(e)= c(\delta^-(X)) - f(\delta^-(X)) + f(\delta^+(X)), \\
\delta_f^+(X) &= \sum_{e\in \delta^+(X)} c_f(e) = c(\delta^+(X)) - f(\delta^+(X)) + f(\delta^-(X)), \\
\delta_f(X,Y) &= \sum_{e\in \delta(X,Y)} c_f(e)=c(\delta(X,Y)) - f(\delta(X,Y)) + f(\delta(Y,X)).
\end{align*}
The next two claims derive simple bounds from \eqref{cond:degree} on the residual in-and out-capacities of cuts.
\begin{claim} \label{cl:residual-term-set}
If $X \subseteq V$, then $\delta_f^-(X) + 2 \sum_{t \in T \cap X} z(t) \ge \delta_f^+(X)$.
\end{claim}
\begin{proof}
\begin{align*}
\delta_f^-(X)-\delta_f^+(X) &= c(\delta^-(X))-c(\delta^+(X))-2\rb{f(\delta^-(X))-f(\delta^+(X))} \\
&=c(\delta^-(X))-c(\delta^+(X))- 2 \sum_{t \in T \cap X} z(t) \ge - 2 \sum_{t \in T \cap X} z(t).
\end{align*}
The equality is by flow conservation. The inequality is by \eqref{cond:degree}. The claim follows.
\end{proof}

\begin{claim} \label{cl:residual-one-term-set}
Consider $X \subseteq V$ such that $X \cap T = \{t_i\}$ for some $1 \le i \le k$. Then $\delta_f^-(X) \ge 1 - z_i$.
\end{claim}
\begin{proof}
\[ \delta_f^-(X) = c(\delta^-(X)) - \rb{f(\delta^-(X)) - f(\delta^+(X))} \ge 1 - z_i. \]
Here we used \eqref{cond:degree} and the flow conservation $f(\delta^-(X)) - f(\delta^+(X)) = z_i$ (as the single sink  contained in $X$ is $t_i$).
\end{proof}

The next claim shows that the large terminals can be separated from the small ones by a cut of residual in-degree $0$. This follows from the lexicographically maximal choice, and is not a particular property of the threshold $1/4$ (it remains true if we replace $\ell$ by any $1 \le j \le k$).
\begin{claim}\label{cl:A}
There exists a set $A\subseteq V$ with $A \cap T = \{ t_1, ..., t_\ell \}$ (i.e., $A$ contains exactly the large terminals) such that $\delta_f^-(A) = 0$, and $\delta_f^+(A)\le 2\sum_{i=1}^\ell z_\ell$.
\end{claim}
\begin{proof}
If $\ell = 0$, then we can choose $A = \emptyset$. So assume $0 < \ell < k$.
Consider the maximum flow in the residual graph $D_f$ from the source set $\{ t_{\ell + 1}, ...,
t_k \}$ to the sink set $\{ t_1, ..., t_{\ell} \}$. If its value is positive,
then there exists a path $P$ in $D_f$ from $t_i$ to $t_j$ for some $i > \ell$
and $j \le \ell$ (without loss of generality it contains no other terminals).
Set $\eps = \min \{ z(t_i), \min_{(u,v) \in P} c_f(u,v) \} > 0$. Then the
$s$-$T$ flow $f' = f + \eps \cdot \one_{P}$ has a lexicographically larger imbalance
sequence than $f$ because $z_i$ is increased without decreasing any other of
the large imbalances, a contradiction. So there must be a cut $A\subseteq V$
with $A\cap T=\{ t_1, ..., t_{\ell} \}$ and $\delta_f^-(A) = 0$.\footnote{Note that $s \notin A$ since $D_f$ contains a path from $t_k$ to $s$.}
The second part follows by the first via \protect \MakeUppercase {C}laim\nobreakspace \ref {cl:residual-term-set}.
\end{proof}

\begin{claim} \label{cl:Bi}
For any $\ell+1 \le i \le k$  (i.e., $t_i$ is a small terminal) there exists a set $B_i \subseteq V \setminus A$ such that $B_i \cap T = \{t_i\}$ and $\delta_f^-(B_i) - \delta_f(A,B_i) < \frac 14 - z_i$.
\end{claim}
\begin{proof}
If $k = \ell + 1$, then we can choose $B_k = V \setminus A$; we have $\delta_f^-(B_k) - \delta_f(A,B_k)=0$ since all arcs leaving the source $s$ are saturated.  So assume $k - \ell \ge 2$. Consider the maximum flow from the source set $\{ t_{\ell + 1}, ..., t_k \} \setminus \{ t_i \}$ to the sink $t_i$ in the graph $D_f \setminus (A\cup\{s\})$.

If its value is at least $\frac 14 - z_i$, then we will get a contradiction by increasing the imbalance of $t_i$ to at least $\frac 14$ without changing any of the large imbalances. Namely, let $g$ be a flow from $\{ t_{\ell + 1}, ..., t_k \} \setminus \{ t_i \}$ to $t_i$ of value $\frac 14 - z_i$. Consider the vector $f + g$. There are two possible cases:
\begin{itemize}
	\item If $f+g$ is still an $s$-$T$ flow, i.e., if for all $j$ we have $g(\delta^+(t_j)) - g(\delta^-(t_j)) \le z(t_j)$, then it has a lexicographically larger imbalance sequence than $f$, a contradiction.
	\item Otherwise pick the maximum $\alpha > 0$ such that $f + \alpha g$ is still an $s$-$T$ flow, i.e., for all $j$ we have $\alpha \rb{g(\delta^+(t_j)) - g(\delta^-(t_j))} \le z(t_j)$, with equality for some $j$. This means that $f + \alpha g$ is an $s$-$T$ flow where at least one terminal $t_j$ has zero imbalance, i.e., it can be removed from the set $T$, contradicting its minimality.
\end{itemize}
So there must be a cut $B_i\subseteq V\setminus (A \cup \{s\})$ such that $B_i\cap T=\{t_i\}$ and
\[\frac 14 - z_i > \delta_f(V\setminus (A\cup\{s\}\cup B_i), B_i) = \delta_f^-(B_i) - \delta_f(A,B_i) -\delta_f(s,B_i). \]
The claim follows by $\delta_f(s,B_i)=0$, which holds since all edges in $\delta^+(s)$ are saturated in $f$.
\end{proof}

The argument uses the bound $\delta_f^+(A)\le 2\sum_{i=1}^\ell z_\ell$ and the fact that all the $B_i$'s must receive a large part of their residual in-degrees from $A$. Since the sets $B_i$ overlap, we have to take their intersections into account. Let us therefore define
\[
\Bs:=\bigcup_{i,j>\ell, i\neq j} (B_i\cap B_j) \subseteq V \setminus (A \cup \{s\})
\]
as the set of vertices contained in at least two sets $B_i$. Let $\alphas:=\delta_f(A,\Bs)$, $\alpha_i:=\delta_f(A,B_i\setminus \Bs)$, and $\beta_i:=\delta_f(\Bs,B_i\setminus \Bs)$
for each $\ell + 1 \le i \le k$.

\begin{claim} \label{cl:alpha-plus-beta}
For each $\ell + 1 \le i \le k$ we have $\frac 34 < \alpha_i + \beta_i$.
\end{claim}
\begin{proof}
Note that $(B_i\setminus \Bs)\cap T = \{t_i\}$, and thus $\delta_f^-(B_i\setminus \Bs)\ge 1-z_i$ by \protect \MakeUppercase {C}laim\nobreakspace \ref {cl:residual-one-term-set}.
From this we can see
\begin{align*}
  1-z_i &\le \delta_f^-(B_i\setminus \Bs) \\ &\le \left(\delta_f^-(B_i)-\delta_f(A,B_i)\right)+\delta_f(A,B_i\setminus \Bs)+\delta_f(\Bs,B_i\setminus \Bs) \\
  &< \frac 14 -z_i+\alpha_i+\beta_i,
\end{align*}
where the second inequality follows because an edge entering $B_i \setminus \Bs$ either enters $B_i$ from outside of $A$, or enters $B_i \setminus \Bs$ from $\Bs$, or enters $B_i \setminus \Bs$ from $A$.
\end{proof}
For the residual in-degree of the set $\Bs$, we apply the trivial bound
\[
\delta^-_f(\Bs) \le \delta_f(A,\Bs) + \sum_{i=\ell+1}^k \rb{\delta^-_f(B_i)-\delta(A,B_i)} \le \alphas + \frac 14 (k-\ell).
\]
The last estimate is by the choice of the sets $B_i$ in \protect \MakeUppercase {C}laim\nobreakspace \ref {cl:Bi}. For the residual out-degree, we get
\[\delta^+_f(\Bs)\ge \sum_{i=\ell+1}^k \beta_i > \frac 34 (k-\ell)- \sum_{i=\ell+1}^k \alpha_i,\]
using \protect \MakeUppercase {C}laim\nobreakspace \ref {cl:alpha-plus-beta}. Applying \protect \MakeUppercase {C}laim\nobreakspace \ref {cl:residual-term-set} to $\Bs$ and noting that $\Bs\cap T=\emptyset$ gives $\delta^-_f(\Bs)\ge \delta^+_f(\Bs)$. Putting \protect \MakeUppercase {C}laim\nobreakspace \ref {cl:A} and the above two bounds together, we conclude that
\begin{equation}\label{eq:f-plus-A}
 2\sum_{i=1}^\ell z_i\ge  \delta_f^+(A)\ge \alphas+\sum_{i=\ell+1}^k \alpha_i\ge \frac 12 (k-\ell).
\end{equation}
Lemma\nobreakspace \ref {lem:not_too_many_small} now follows.
\end{proof}

\fi

\section{Algorithm for Local-Connectivity ATSP}
\label{sec:lcATSP}

We prove our main result in this section. Our claim for ATSP follows from solving Local-Connectivity ATSP:

\begin{theorem} \label{thm:main}
There is a polynomial-time  $100$-light algorithm for Local-Connectivity ATSP on
graphs with two edge weights.
\end{theorem}
Together with Theorem\nobreakspace \ref {thm:LoctoGlo}, this implies our main result:
\begin{theorem}
For any graph with two edge weights, the integrality gap of its Held-Karp relaxation is at most $500$. Moreover, we can find an $901$-approximate tour in polynomial time.
\end{theorem}
The factor $500$ comes from $5\cdot 100$. In Theorem\nobreakspace \ref {thm:LoctoGlo}, we
select $\varepsilon$ such that 
$(9+\varepsilon)\cdot 100 \leq 901$.
Our proof of Theorem\nobreakspace \ref {thm:main} proceeds as outlined in Section\nobreakspace \ref {sec:overview}. 
\iflncs
\else 
In Section\nobreakspace \ref {sec:splitgr}, we give an algorithm for calculating $\lb$
and define the \emph{split graph} which will be central for finding light
solutions. In Section\nobreakspace \ref {sec:solvlcATSP}, we then show how to use these concepts
to solve Local-Connectivity ATSP for any given partitioning of the vertices.
\fi

Recall that the edges are partitioned into the set $E_0$ of cheap edges and the
set $E_1$ of expensive edges.  Set $\xs$ to be an optimal solution to the
Held-Karp relaxation. We start by noting that the problem is easy if $\xs$
assigns very small total fractional value to expensive edges. In that case, we can
easily reduce the problem to the unweighted case which was solved
in~\cite{Svensson15}.   

\begin{lemma} \label{lem:expensive_less_than_one_is_easy}
There is a polynomial-time $6$-light algorithm for Local-Connectivity ATSP for graphs where $\xs(E_1) < 1$.
\end{lemma}

\iflncs\else
\begin{proof}
If $\xs(E_1) = 0$, then just apply the standard $3$-light polynomial-time algorithm for unweighted graphs \cite{Svensson15}. So suppose that $0 < \xs(E_1) < 1$. Then clearly the graph $(V,E_0)$ is strongly connected, i.e. every pair of vertices is connected by a directed path of cheap edges (of length at most $n-1$). Thus each expensive edge $(u,v)$ can be replaced by such a $u$-$v$-path $P(u,v)$. Let us obtain a new circulation $x'$ from $\xs$ by replacing all expensive edges in this way, i.e.,
\[ x' = \xs|_{E_0} + \sum_{(u,v) \in E_1} \xs_{(u,v)} \cdot \one_{P(u,v)}. \]
To bound the cost of $x'$, note that $\xs(E_0) = \xs(E) - \xs(E_1) > n-1 > (n-1)\xs(E_1)$ and thus
\[ w(x') \le w_0 \cdot \xs(E_0) + (n-1) w_0 \cdot \xs(E_1) \le 2 w_0 \cdot \xs(E_0) \le 2 w(\xs). \]
By construction, $x'$ is a feasible solution for the Held-Karp relaxation and $\supp(x') \subseteq E_0$. Therefore we can use it in the standard $3$-light polynomial-time algorithm for the unweighted graph $(V,E_0)$. Together with the bound $w(x') \le 2w(\xs)$ this gives a $6$-light algorithm.
\end{proof}
\fi

For the rest of this section, we thus assume $\xs(E_1) \ge 1$. Our
objective is to define a function $\lb : V \to \bR_+$ such that $\lb(V) \le
\OPT = w(\xs)$ and then show how to, given a partition $V = V_1 \cup ...
\cup V_k$, find an Eulerian set of edges $F$ which crosses all $V_i$-cuts and
is $\cO(1)$-light with respect to the defined $\lb$ function.


\subsection{Calculating $\lb$ and constructing the split graph}
\label{sec:splitgr}
\iflncs \else
First, we use our flow decomposition technique to find a small set of
\emph{terminals} $T$ such that it is possible to route a certain flow $f$ from
endpoints of all expensive edges to $T$. Next, we use $f$ and $T$ to 
calculate the function $\lb$ and to construct a \emph{split graph} $\gsp$, where each
vertex of $G$ is split into two.\fi

\paragraph{Finding terminals $T$ and flow $f$.} 
\iflncs
\else
We use Theorem\nobreakspace \ref {thm:flow_thm} to obtain a small-enough set of terminals $T$ and
a flow $f$ which takes all flow on expensive edges to this set $T$. More
precisely, we have the following corollary of Theorem\nobreakspace \ref {thm:flow_thm}.
\fi
\begin{lemma}\label{lem:flow_app}
There exist a vertex set $T \subseteq V$ and a flow $f : E \to \bR_+$ from
source set $\{ \tail(e) : e \in E_1 \}$ to sink set $T$ of value $\xs(E_1)$
such that:
\iflncs
\else
\begin{itemize}\itemsep0mm
	\item $|T| \le 8 \xs(E_1)$,
	\item $f \le \xs$,
	\item $f$ saturates all expensive edges, i.e., $f(e) = \xs_e$ for all $e \in E_1$,
	\item for each $t \in T$, $f(E_0 \cap \delta^+(t)) = 0$ and $f(\delta^-(t)) > 0$.
\end{itemize}
\fi
Moreover, $T$ and $f$ can be computed in polynomial time.
\end{lemma}
\iflncs\else
\begin{proof}
We construct $G'$ to be $G$ with a new vertex $s$, where the tail of every expensive edge is redirected to be $s$. Formally, $V(G') = V \cup \{s\}$ and $E(G') = E_0 \cup \{ (s,\head(e)) : e \in E_1 \}$. The capacity vector $c$ is obtained from $\xs$ by just following this redirection, i.e., for any edge $e' \in E(G')$ we define $c(e') = \xs_{e}$, where $e \in E(G)$ is taken to be the preimage of $e'$ in $G$.

Clearly $c(\delta^+(s)) = \xs(E_1)$, and $s \in V(G')$ has no incoming edges. To see that condition~\eqref{cond:degree} of~Theorem\nobreakspace \ref {thm:flow_thm} is satisfied, recall that for every $\emptyset \ne S \subsetneq V(G)$ we have $\xs(\delta^+(S)) = \xs(\delta^-(S)) \ge 1$; redirecting the tail of some edges to $s$ can only reduce the outdegree or increase the indegree of $S$, i.e., $c(\delta^-(S)) \ge \xs(\delta^-(S)) = \xs(\delta^+(S)) \ge c(\delta^+(S))$. This gives condition~\eqref{cond:degree} for all sets $S \subsetneq V(G') - s$; for $S = V(G') - s$, note that $c(\delta^-(S)) = \xs(E_1) \ge 1 = \max\{1, 0\} = \max\{1, c(\delta^+(S))\}$ since we assumed that $\xs(E_1) \ge 1$.

From Theorem\nobreakspace \ref {thm:flow_thm} we obtain a vertex set $T$ with $|T| \le 8\xs(E_1)$ and a flow $f' : E(G') \to \bR_+$ from $s$ to $T$ of value $\xs(E_1)$ with $f' \le c$. We can assume $f'(\delta^+(t))=0$ for all $t\in T$: in a path-cycle decomposition of $f'$ we can remove all cycles and
terminate every path at the first terminal it reaches.
The flow $f$ is obtained by mapping $f'$ back to $G$, i.e., taking each $f(e)$ to be $f'(e')$, where $e'$ is the image of $e$. Note that $f'$ must saturate all outgoing edges of $s$, so $f$ saturates all expensive edges. For the last condition, the part $f(E_0 \cap \delta^+(t)) = 0$ is implied by $f'(\delta^+(t)) = 0$, and for the part $f(\delta^-(t)) > 0$, note that if $f(\delta^-(t)) = 0$, then we could have removed $t$ from $T$.

Note that such a set $T$ can be found in polynomial time: starting from $T=V$ (for which the required flow exists: consider $f' = c|_{\delta^+(s)}$, the restriction of $c$ to ${\delta^+(s)}$), we remove vertices from $T$ one by one until we obtain a minimal set $T$ such that there exists a flow of value $\xs(E_1)$ from $s$ to $T$. 
\end{proof}
\fi

\paragraph{Definition of $\lb$.} We set $\lb : V \rightarrow \bR_+$ to be a scaled-down variant of $\lbs: V \rightarrow \bR_+$ which is defined as follows:
\[
\lbs(v):=
\begin{cases}
w_0 \cdot \xs(\delta^-(v))                                   & \mbox{if } v \notin T,\\
w_0 \cdot \xs(\delta^-(v)) + w_1 \cdot \ceil{f(\delta^-(t))} & \mbox{if } v \in T.
\end{cases}
\]
The definition of $\lb$ is now simply $\lb(v) = \lbs(v)/10$. The scaling-down is done so as to satisfy $\lb(V) \leq \OPT$ (see Lemma\nobreakspace \ref {lem:sum_of_lbs_is_low}).
Clearly we have $\lbs(v) \ge w_0$ for all $v \in V$ and $\lbs(t) \ge w_1 + w_0 \ge w_1$ for terminals $t \in T$.

The intuition behind this setting of $\lbs$ is that we want to pay for each
expensive edge $e \in E_1$ in the terminal $t \in T$ which the flow $f$
``assigns'' to $e$. Indeed, in the split graph we will reroute flow (using $f$)
so as to ensure that any path which traverses $e$ must then visit such
a terminal $t$ to offset the cost of the expensive edge.
\iflncs \else As for the total cost
of $\lbs$, note that if we removed the rounding from its definition, then we
would get $\lbs(V) \le 2 \OPT$, since
\[ w_0 \cdot \sum_{v \in V} \xs(\delta^-(v)) + w_1 \cdot \sum_{v \in T} f(\delta^-(t)) = w_0 \cdot \xs(E_0) + w_0 \cdot \xs(E_1) + w_1 \cdot \xs(E_1) \le 2 w(\xs) \]
(here we used that $f$ is of value $\xs(E_1) = \sum_{t \in T} f(\delta^-(t))$).
So, similarly to the $3$-light algorithm for unweighted metrics
in~\cite{Svensson15}, the key is to argue that rounding does not increase this
cost too much. For this, we will take
advantage of the small size of $T$. Details are found in the proof of the
following lemma.
\fi
\begin{lemma} \label{lem:sum_of_lbs_is_low}
$\lbs(V) \le 10 \cdot \OPT.$
\end{lemma}
\iflncs
\else
\begin{proof}
  The bound follows from elementary calculations:
\begin{align*}
\lbs(V) &= w_0 \cdot \sum_{v \in V} \xs(\delta^-(v)) + w_1 \cdot \sum_{t \in T} \ceil{f(\delta^-(t))} \\
      &\le w_0 \cdot \sum_{v \in V}  \xs(\delta^-(v)) + w_1 \cdot \sum_{t \in T} \rb{f(\delta^-(t)) + 1} \\
      &\le  w_0 \cdot \xs(E) + w_1 \cdot \rb{\xs(E_1) + |T|} \\
      &\le  w(\xs) + w_1 \cdot 9 \xs(E_1) \\
      &\le 10 w(\xs)
\end{align*}
(recall that $|T| \le 8 \xs(E_1)$ by Lemma~\ref{lem:flow_app}).
\end{proof}
\fi

\paragraph{Construction of the split graph.} The next step is to reroute flow so as
to ensure that all expensive edges are ``paid for'' by the $\lb$ at terminals.
To this end, we define a new \emph{split graph} and a \emph{split circulation}
on \iflncs\else it (see also Fig.\nobreakspace \ref {fig:gspex} for an example). \fi
\begin{definition} \label{def:split}
The \emph{split graph} $\gsp$ is defined as follows. For every $v \in V$ we create two copies $v^0$ and $v^1$ in $V(\gsp)$. For every cheap edge $(u,v) \in E_0$:
\begin{itemize}
	\item if $\xs(u,v) - f(u,v) > 0$, create an edge $(u^0,v^0)$ in $E(\gsp)$ with $\xsp(u,v) = \xs(u,v) - f(u,v)$,
	\item if $f(u,v) > 0$, create an edge $(u^1,v^1)$ in $E(\gsp)$ with $\xsp(u,v) = f(u,v)$.
\end{itemize}
  For every expensive edge $(u,v) \in E_1$ we create one edge $(u^0,v^1)$ in $E(\gsp)$ with $\xsp(u,v) = f(u,v)$. Finally, for each $t \in T$ we create an edge $(t^1, t^0)$ in $E(\gsp)$ with $\xsp(t^1, t^0) = f(\delta^-(t))$.
  
  The new edges are weighted as follows: images of edges in $E_0$ have weight $w_0$, the images of edges in $E_1$ have weight $w_1$, and the new edges $(t^1,t^0)$ have weight $0$. Let us denote the new weight function by $\wsp$.
  
  Vertices $v^0$ will be called \emph{free vertices} and vertices $v^1$ will be called \emph{debt vertices}. Edges entering a free vertex will be called \emph{free edges}, and those entering a debt vertex will be called \emph{debt edges}.
\end{definition}
\iflncs
\else
A fundamental consequence of our construction is the following.
\begin{fact} \label{fact:debt-path}
Consider any path $P$ in $\gsp$ such that the first edge of $P$ is a debt edge and the last one is a free edge or an expensive edge. Then $P$ must go through a terminal, i.e., it must contain an edge $(t^1,t^0)$ for some $t \in T$. \qedmanual
\end{fact}
We also have the following two properties by design.
\begin{fact} \label{fact:xsp_is_circulation}
The vector $\xsp : E(\gsp) \to \bR_+$ is a circulation in $\gsp$. 
\end{fact}
\begin{proof}
 Intuitively, this is because our rerouting corresponds to taking every path
 $v_0, v_1, ..., v_k$ in a path-cycle decomposition of $f$ (where $(v_0, v_1)$
 is an expensive edge and $v_k \in T$ is a terminal) and, rather than placing it
 on the ``free level'' in $\gsp$ (i.e., mapping it to $v_0^0, v_1^0, ...,
 v_k^0$), instead routing it as follows: $v_0^0, v_1^1, v_2^1, ..., v_{k-1}^1,
 v_k^1, v_k^0$. We can think of the edge $(v_0^0, v_1^1)$ as incurring a debt,
 and of the edge $(v_k^1, v_k^0)$ as discharging this debt.
 
 Now we proceed to give a formal proof.
 For all $v \in V$ we will prove that $\xsp(\delta^+(v^0)) = \xsp(\delta^-(v^0))$ and $\xsp(\delta^+(v^1)) = \xsp(\delta^-(v^1))$.

Suppose $v = t \in T$. Then $\xsp(\delta^-(t^1)) = f(\delta^-(t))$. Also $\xsp(\delta^+(t^1)) = \xsp(t^1,t^0) = f(\delta^-(t))$ because of the property that $f(E_0 \cap \delta^+(t)) = 0$. For $t^0$ we have
\begin{align*}
\xsp(\delta^-(t^0)) &= f(\delta^-(t)) + (\xs - f)(E_0 \cap \delta^-(t)) \\ &= f(E_1 \cap \delta^-(t)) + \xs(E_0 \cap \delta^-(t)) \\&= \xs(\delta^-(t))
\end{align*}
and
\begin{align*}
\xsp(\delta^+(t^0)) &= (\xs - f)(E_0 \cap \delta^+(t)) + \xs(E_1 \cap \delta^+(t)) \\ &= \xs(E_0 \cap \delta^+(t)) + \xs(E_1 \cap \delta^+(t)) \\ &= \xs(\delta^+(t))
\end{align*}
where the second inequality again follows by $f(E_0 \cap \delta^+(t)) = 0$.
This implies that $\xsp(\delta^-(t^0)) = \xs(\delta^-(t)) = \xs(\delta^+(t)) = \xsp(\delta^+(t^0))$.

We turn to the case $v \notin T$. Note that since all expensive edges are saturated by $f$, and their tails are exactly the sources of $f$, we have
\begin{equation}
f(E_0 \cap \delta^+(v)) = f(E_0 \cap \delta^-(v)) + \xs(E_1 \cap \delta^-(v)). \label{eq:flow_conservation}
\end{equation}
Since the left hand side is equal to $\xsp(\delta^+(v^1))$ and the right hand side is $\xsp(\delta^-(v^1))$, we have proved our claim for $v^1$. For $v^0$, note that the incoming edges of $v_0$ are all cheap, and the incoming flow is
\[ \xsp(\delta^-(v^0)) = (\xs - f)(E_0 \cap \delta^-(v)). \]
As for outgoing flow, $v^0$ is also the tail of all expensive edges whose tail in $G$ was $v$, i.e.,
\begin{align*}
\xsp(\delta^+(v^0)) &= (\xs - f)(E_0 \cap \delta^+(v)) + \xs(E_1 \cap \delta^+(v)) \\
                    &= \xs(\delta^+(v)) - f(E_0 \cap \delta^+(v)) \\
                    &= \xs(\delta^-(v)) - f(E_0 \cap \delta^-(v)) - \xs(E_1 \cap \delta^-(v)) \\
                    &= \xs(E_0 \cap \delta^-(v)) - f(E_0 \cap \delta^-(v)) \\
                    &= (\xs - f)(E_0 \cap \delta^-(v)) \\
                    &= \xsp(\delta^-(v^0)),
\end{align*}
where the third equality is by~\eqref{eq:flow_conservation}.
\end{proof}


\begin{fact} \label{fact:at_least_1_unit}
For each proper subset $U \subset V$ we have
$\xsp(\delta^-(\{v^0,v^1 : v \in U\}))
 = \xsp(\delta^+(\{v^0,v^1: v \in U\})) \ge 1$. In other words, the image of
 every cut in $G$ is still crossed by at least one  unit  of $\xsp$.
\end{fact}
\begin{proof}
 This follows since $\xs(\delta^+(U)) = \xs(\delta^-(U)) \ge 1$, and contracting
 all pairs $v^0, v^1$ would yield back $G$ and $\xs$.
\end{proof}
\fi

\iflncs\else
\begin{figure}[t]
  \centering
  \begin{tikzpicture}
  \node at (1.25, 1) {\footnotesize $G$ and $f$};
  \foreach \val/\nam/\namm in {0/a/d,1/b/e,2/c/g} {
    \ifthenelse{\val=2}
    {
      \node[sgvertex,fill=black] (\nam) at (0,-1.5*\val) {\scriptsize \color{white} $\nam$};
      \node[sgvertex, fill=black] (\namm) at (2.5,-1.5*\val) {\scriptsize\color{white} $\namm$};
    }
    {
      \node[sgvertex] (\nam) at (0,-1.5*\val) {\scriptsize $\nam$};
      \node[sgvertex] (\namm) at (2.5,-1.5*\val) {\scriptsize $\namm$};
    }
  }
  \draw (a) edge[->] node[fill=white, inner sep=2pt] {\tiny $\nicefrac{1}{3}$} (b);
  \draw (b) edge[->] node[fill=white, inner sep=2pt] {\tiny $\nicefrac{2}{3}$} (c);
  \draw (c) edge[->, bend left=40] node[fill=white, inner sep=2pt] {\tiny $0$} (a);
  \draw (d) edge[->] node[fill=white, inner sep=2pt] {\tiny $\nicefrac{1}{3}$} (e);
  \draw (e) edge[->] node[fill=white, inner sep=2pt] {\tiny $\nicefrac{2}{3}$} (g);
  \draw (g) edge[->, bend right=40] node[fill=white, inner sep=2pt] {\tiny $0$} (d);
  
  \draw (a) edge[->, bend left=20, thick] node[fill=white, inner sep=2pt] {\tiny $\nicefrac{1}{3}$} (d);
  \draw (d) edge[->, bend left=20, thick] node[fill=white, inner sep=2pt] {\tiny $\nicefrac{1}{3}$} (a);
  
  \draw (b) edge[->, bend left=20, thick] node[fill=white, inner sep=2pt] {\tiny $\nicefrac{1}{3}$} (e);
  \draw (e) edge[->, bend left=20, thick] node[fill=white, inner sep=2pt] {\tiny $\nicefrac{1}{3}$} (b);
  
  \draw (c) edge[->, bend left=20, thick] node[fill=white, inner sep=2pt] {\tiny $\nicefrac{1}{3}$} (g);
  \draw (g) edge[->, bend left=20, thick] node[fill=white, inner sep=2pt] {\tiny $\nicefrac{1}{3}$} (c);

  \begin{scope}[xshift=6cm]
  \node at (1.9, 1) {\footnotesize $\gsp$ and $\xsp$};
  \foreach \val/\nam/\namm in {0/a/d,1/b/e,2/c/g} {
    \ifthenelse{\val=2}
    {
      \node[sgvertex, fill=black] (\nam0) at (-0.15,-1.5*\val+0.3) {\color{white} \tiny $\nam^0$};
      \node[sgvertex,rectangle, fill=black] (\nam1) at (0.75,-1.5*\val-0.3) {\color{white} \tiny $\nam^1$};
      \node[sgvertex,rectangle, fill=black] (\namm1) at (3,-1.5*\val-0.3) {\color{white}\tiny $\namm^{1}$};
      \node[sgvertex, fill=black] (\namm0) at (3.9,-1.5*\val+0.3) {\color{white}\tiny $\namm^{0}$};
    }
    {
      \node[sgvertex] (\nam0) at (-0.15,-1.5*\val+0.3) {\tiny $\nam^0$};
      \node[sgvertex,rectangle] (\nam1) at (0.75,-1.5*\val-0.3) {\tiny $\nam^1$};
      \node[sgvertex,rectangle] (\namm1) at (3,-1.5*\val-0.3) {\tiny $\namm^{1}$};
      \node[sgvertex] (\namm0) at (3.9,-1.5*\val+0.3) {\tiny $\namm^{0}$};
    }
  }
  \draw (a0) edge[->, bend left=20, thick] node[fill=white, inner sep=2pt, pos=0.81] {\tiny $\nicefrac{1}{3}$} (d1);
  \draw (b0) edge[->, bend left=20, thick]node[fill=white, inner sep=2pt, pos=0.81] {\tiny $\nicefrac{1}{3}$} (e1);
  \draw (c0) edge[->, bend left=20, thick]node[fill=white, inner sep=2pt, pos=0.81] {\tiny $\nicefrac{1}{3}$} (g1);
  \draw (d0) edge[->, bend right=20, thick]node[fill=white, inner sep=2pt, pos=0.81] {\tiny $\nicefrac{1}{3}$} (a1);
  \draw (e0) edge[->, bend right=20, thick]node[fill=white, inner sep=2pt, pos=0.81] {\tiny $\nicefrac{1}{3}$} (b1);
  \draw (g0) edge[->, bend right=20, thick]node[fill=white, inner sep=2pt, pos=0.81] {\tiny $\nicefrac{1}{3}$} (c1);
  \draw (a0) edge[->]node[fill=white, inner sep=2pt] {\tiny $\nicefrac{1}{3}$} (b0);
  \draw (d0) edge[->]node[fill=white, inner sep=2pt] {\tiny $\nicefrac{1}{3}$} (e0);
  \draw (a1) edge[->]node[fill=white, inner sep=2pt, pos=0.25] {\tiny $\nicefrac{1}{3}$} (b1);
  \draw (b1) edge[->]node[fill=white, inner sep=2pt, pos=0.25] {\tiny $\nicefrac{2}{3}$} (c1);
  \draw (d1) edge[->]node[fill=white, inner sep=2pt, pos=0.25] {\tiny $\nicefrac{1}{3}$} (e1);
  \draw (e1) edge[->]node[fill=white, inner sep=2pt, pos=0.25] {\tiny $\nicefrac{2}{3}$} (g1);
  \draw (c0) edge[->, bend left=40]node[fill=white, inner sep=2pt, pos=0.4] {\tiny $\nicefrac{2}{3}$} (a0);
  \draw (g0) edge[->, bend right=40]node[fill=white, inner sep=2pt, pos=0.4] {\tiny $\nicefrac{2}{3}$} (d0);
  \draw (c1) edge[->, bend left=45]node[fill=white, inner sep=2pt, pos=0.4] {\tiny ${1}$} (c0);
  \draw (g1) edge[->, bend right=45]node[fill=white, inner sep=2pt, pos=0.4] {\tiny ${1}$} (g0);
  \end{scope}
\end{tikzpicture}
  \caption{An example of the construction of $\gsp$ and $\xsp$ from $G, \xs$ and $f$. Here $\xs$ is $1/3$ for the expensive edges (depicted as thick) and $2/3$ for the remaining (cheap) edges; the set $T$ of terminals of the flow $f$ is depicted in black.}
  \label{fig:gspex}
\end{figure}
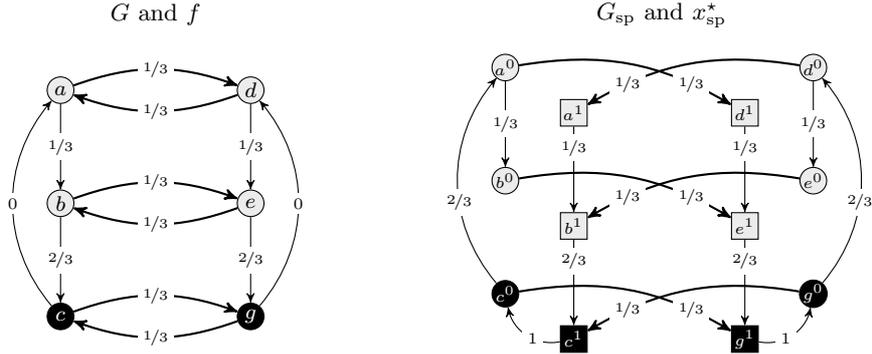
\fi

\subsection{Solving Local-Connectivity ATSP}
\label{sec:solvlcATSP}
Now our algorithm is given a partition $V = V_1 \cup ... \cup V_k$ of
the original vertex set. The
objective is to output a set of edges $F$ which crosses all $V_i$-cuts and is
$\cO(1)$-light with respect to our $\lb$ function.

We are aiming for a similar construction as in the unit-weight case:
based on the split graph $\gsp$, we construct an integer circulation
problem with an auxiliary vertex $A_i$ representing a certain subset
$U_i\subseteq V_i$ for every $1\le i\le k$. We then map its solution
back to the original graph and patch up the loose endpoints inside
every $U_i$ by a path. However, we have to account for the following
difficulties: {\em (i)} an edge leaving $U_i$ should also leave $V_i$;
{\em (ii)} debt should not disappear inside $U_i$: if the edge
entering it carries debt but the edge leaving does not, we must make
sure this difference can be charged to a terminal in $U_i$; {\em
  (iii)} the path used inside $U_i$ must pay for all expensive edges
it uses. 

All three issues can be appropriately tackled by defining an \emph{auxiliary graph} inside $V_i$. Edges of the auxiliary graph represent paths containing one expensive edge and one terminal (which can pay for themselves); however, these paths may not map to paths in the split graph. We select the subset $U_i\subseteq V_i$ as a sink component in the auxiliary graph. 
 For convenience,  Figs.\nobreakspace \ref {tab:flows} and\nobreakspace  \ref {fig:overview_lcatsp} give an overview of the different steps, graphs and flows used by our algorithm.
 
\iflncs\else
\begin{figure}[t]
	\centering
	\begin{tabular}{lllll}
		name & graph & circulation & integral & obtained from the previous by \\
		\hline
		$\xs$ & $G$ & yes & no & solving LP (Section\nobreakspace \ref {sec:prelim}) \\
		$\xsp$ & $\gsp$ & yes & no & splitting (Definition\nobreakspace \ref {def:split}) \\
		$\xspp$ & $\gsp'$ & yes & no & redirecting edges to $A_i$ \\
		$\yspp$ & $\gsp'$ & yes & yes & rounding to integrality (Lemma\nobreakspace \ref {lem:exists_integral_y}) \\
		$\ysp$ & $\gsp$ & no & yes & redirecting edges back from $A_i$ \\
		$y$ & $G$ & no & yes & mapping back to $G$ \\
		$F$ & $G$ & yes & yes & adding walks $P_i$
	\end{tabular}
	\caption{This table summarizes the various circulations and pseudo-flows that appear in our algorithm, in order.}
	\label{tab:flows}
\end{figure}
\fi

\paragraph{Construction of auxiliary graphs and modification of split graph.}
Our first step is to construct an \emph{auxiliary graph} for each component
$V_i$. The strong-connectivity structure of this graph will guide our algorithm.
\begin{definition} \label{def:auxiliary_graph}
The \emph{auxiliary graph} $\giaux$ is a graph with vertex set $V_i$ and the following edge set: for $u, v \in V_i$, $(u,v) \in E(\giaux)$ if any of the following three conditions is satisfied:
\begin{itemize}\itemsep0mm
	\item there is a cheap edge $(u,v) \in E_0 \cap G[V_i]$ inside $V_i$, or
	\item there is a $u$-$v$-path in $G[V_i]$ whose first edge is expensive and all other edges are cheap, and $v \in T$ is a terminal -- we then call the edge $(u,v) \in E(\giaux)$ a \emph{postpaid edge} -- or
	\item there is a $u$-$v$-path in $G[V_i]$ whose last edge is expensive and all other edges are cheap, and $u \in T$ is a terminal -- we then call the edge $(u,v) \in E(\giaux)$ a \emph{prepaid edge}.
\end{itemize}
Define the \emph{preimage} of such an edge $(u,v) \in E(\giaux)$ to be a \emph{shortest} path inside $V_i$ as above (in the first case, a single edge).
\end{definition}

Now, for each $i$  consider a decomposition of $\giaux$ into strongly connected components.\footnote{Note that we decompose the vertex set $V_i$, but with respect to the edge set $E(\giaux)$, not $E(G[V_i])$.}  Let $U_i \subseteq V_i$ be the vertex set of a sink component in this decomposition. That is, there is no edge from $U_i$ to $V_i \setminus U_i$ in the auxiliary graph $\giaux$.
Note that $\giaux$ is constructed based only on the original graph $G$ and not the split graph $\gsp$. However, we will solve Local-Connectivity ATSP by solving an integral circulation problem on $\gsp'$: a modification of the split graph $\gsp$, described as follows.

For each $i$, define $\Uisp = \{ v^0, v^1 : v \in U_i \} \subseteq V(\gsp)$ to
be the set of vertices in the split graph corresponding to $U_i$. (Note that
$\Uisp$ may not be strongly connected in $\gsp$.) We are going to reroute part
of the $\xsp$ flow going in and out of $\Uisp$ to a new auxiliary vertex $A_i$.
While the $3$-light algorithm for unit-weight graphs rerouted flow from all boundary edges of a component $U_i$ (see Section\nobreakspace \ref {sec:overview}), here we will be more careful and choose only a subset of boundary edges of $\Uisp$ to be rerouted.

To this end, select  a subset of edges $X_i^- \subseteq \delta^-(\Uisp)$  with
$\xsp(X_i^-)= 1/2$ such that either all edges in $X_i^-$ are debt edges, or all
are free edges. \iflncs\else This is possible since $\xsp(\Uisp)\ge 1$ by
\protect \MakeUppercase {F}act\nobreakspace \ref {fact:at_least_1_unit}.\footnote{To obtain exactly $1/2$, we might need to
break an edge up into two copies, dividing its $\xsp$-value between them appropriately, and include one copy in $X_i^-$ but not the other; we omit this for simplicity of notation, and
assume there is such an edge set with exactly $\xsp(X_i^-)= 1/2$.}\fi

We define the set of outgoing edges $X_i^+ \subseteq \delta^+(\Uisp)$ to be,
intuitively, the edges over which the flow that entered $\Uisp$ by $X_i^-$
exits $\Uisp$.
That is, consider an arbitrary cycle decomposition of the circulation $\xsp$,
and look at the set of cycles containing the edges in $X_i^-$. We define
$X_i^+$ as the set of edges on these cycles that first leave $\Uisp$ after
entering $\Uisp$ on an edge in $X_i^-$; clearly,
$\xsp(X_i^+)=1/2$.\footnote{Again, we might need to split some edges into two
copies.}

Let $g_i$ denote the flow on these cycles connecting the heads of edges in
$X_i^-$ and the tails of edges in $X_i^+$.
We will use the following claim later in the construction.

\begin{fact} \label{fact:backtrack}
Assume all edges in $X_i^-$ are debt edges but $e \in X_i^+$ is a free edge or an expensive edge. Then there exists a path in $\gsp[U_i]$ between a vertex $t^0$ (for some terminal $t \in T$) and the tail of $e$, made up of only cheap edges.
\end{fact}
\iflncs\else
\begin{proof}
Consider the cycle in the cycle decomposition that contains $e$; it enters $U_i$ on a debt edge. Using \protect \MakeUppercase {F}act\nobreakspace \ref {fact:debt-path}, this cycle fragment must contain an edge of the form $(t^1,t^0)$; pick the last such edge. All edges that follow are free and cheap.
\end{proof}\fi

We now transform $\gsp$ into a new graph $\gsp'$ and $\xsp$ into new circulation $\xspp$ as follows. For every set $V_i$ in the partition we introduce a new auxiliary vertex $A_i$
and redirect all edges in $X_i^-$ to point to $A_i$ and those in $X_i^+$ to point from $A_i$. We further subtract the flow $g_i$ inside $\Uisp$; hence the resulting vector $\xspp$ will be a circulation, with $\xspp(\delta^-(A_i))=1/2$.
If $X_i^-$ is a set of free edges, then we will say that $A_i$ is a free vertex, otherwise we say that it is a debt vertex.

\iflncs\else
\begin{figure}[t]
  \centering
    \begin{tikzpicture}
    \node at (1.5, 1.5) {\scriptsize $G[V_i]$};
    \node[ssgvertex, fill=black] (a) at (0,0) {};
    \node[ssgvertex] (b) at (0.9,-0.8) {};
    \node[ssgvertex] (c) at (2.1,-0.8) {};
    \node[ssgvertex] (d) at (3,0) {};
    \node[ssgvertex] (e) at (2.1,0.8) {};
    \node[ssgvertex] (f) at (0.9,0.8) {};

    \draw (a) edge[->, bend right] (b);
    \draw (a) edge[<-, bend left] (b);
    \draw (b) edge[->, thick] (c);
    \draw (c) edge[->, bend right] (d);
    \draw (d) edge[->, thick, bend right] (e);
    \draw (e) edge[->] (f);
    \draw (f) edge[->, bend right] (a);

    \begin{scope}[xshift=5.5cm]
      \node at (1.5, 1.5) {\scriptsize $\giaux$};
      \draw[dashed] (-0.4, -1.1) rectangle (3.3, 0.5);
      \node at (3.04, -0.85) {\scriptsize $U_i$};
      \node[ssgvertex, fill=black] (a) at (0,0) {};
      \node[ssgvertex] (b) at (0.9,-0.8) {};
      \node[ssgvertex] (c) at (2.1,-0.8) {};
      \node[ssgvertex] (d) at (3,0) {};
      \node[ssgvertex] (e) at (2.1,0.8) {};
      \node[ssgvertex] (f) at (0.9,0.8) {};

      \draw (a) edge[->, bend right] (b);
      \draw (a) edge[<-, bend left] (b);
      \draw (a) edge[->, bend left=20, snake it]  (c);
      \draw (c) edge[->, bend right] (d);
      \draw (d) edge[->,  snake it, bend right=20] (a);
      \draw (e) edge[->] (f);
      \draw (f) edge[->, bend right] (a);
    \end{scope}
    \begin{scope}[xshift=0cm,yshift=-4.5cm]
      \node at (1.5, 1.5) {\scriptsize $\gsp$ and $\xsp$ on $V_i$};
      \node[ssgvertex, fill=black] (a0) at (0,-0.3) {};
      \node[ssgvertex, fill=black, rectangle] (a1) at (0,0.3) {};
      \node[ssgvertex] (b) at (0.9,-0.8) {};
      \node[ssgvertex, rectangle] at (0.9,-0.5) {};
      \node[ssgvertex,rectangle] (c1) at (2.1,-0.5) {};
      \node[ssgvertex] (c) at (2.1,-0.8) {};
      \node[ssgvertex] (d) at (3,0) {};
      \node[ssgvertex,rectangle]  at (2.7,0) {};
      \node[ssgvertex,rectangle] (e) at (2.1,0.8) {};
      \node[ssgvertex]  at (2.1,1.1) {};
      \node[ssgvertex, rectangle] (f) at (0.9,0.8) {};
      \node[ssgvertex]  at (0.9,1.1) {};

      \draw (a0) edge[->, bend right=20]node[fill=white, inner sep=1pt, pos=0.5] {\tiny $.5$} (b);
      \draw (a0) edge[<-, bend left=20]node[fill=white, inner sep=1pt, pos=0.5] {\tiny $.5$} (b);
      \draw (b) edge[->, thick]node[fill=white, inner sep=1pt, pos=0.4] {\tiny $.5$} (c1);
      \draw (c) edge[->, bend right]node[fill=white, inner sep=1pt, pos=0.5] {\tiny $1$} (d);
      \draw (d) edge[->, thick, bend right]node[fill=white, inner sep=1pt, pos=0.5] {\tiny $1$} (e);
      \draw (e) edge[->]node[fill=white, inner sep=1pt, pos=0.5] {\tiny $1$} (f);
      \draw (f) edge[->, bend right]node[fill=white, inner sep=1pt, pos=0.5] {\tiny $1$} (a1);
      \draw (a1) edge[->]node[left,  inner sep=1pt, pos=0.4] {\tiny $1$} (a0);

      \draw (a0) edge[->, dashed]node[fill=white, inner sep=1pt, pos=0.6] {\tiny $.5$} (-0.2, -1.6);
      \draw (a0) edge[->, dashed]node[fill=white,inner sep=2pt, pos=0.4] {\tiny $.5$} (-0.7, -1.6);
      \draw (b) edge[<-, dashed]node[fill=white, inner sep=1pt, pos=0.5] {\tiny $.5$} (0.8, -1.6);
      \draw (c1) edge[->, dashed]node[fill=white, inner sep=1pt, pos=0.5] {\tiny $.5$} (3, -1.4);
      \draw (c) edge[<-, dashed]node[fill=white, inner sep=1pt, pos=0.5] {\tiny $.5$} (1.8, -1.6);
      \draw (c) edge[<-, dashed]node[fill=white, inner sep=1pt, pos=0.5] {\tiny $.5$} (2.4, -1.6);
      \node at (1.8, -1.8) {\tiny $X_i^-$};
      \node at (-0.7, -1.8) {\tiny $X_i^+$};
    \end{scope}
    \begin{scope}[xshift=5.5cm, yshift=-4.5cm]
      \node at (1.2, 1.5) {\scriptsize $\gsp'$ and $\xspp$ on $V_i$};
      \node[ssgvertex, fill=black] (a0) at (0,-0.3) {};
      \node[ssgvertex, fill=black, rectangle] (a1) at (0,0.3) {};
      \node[ssgvertex] (b) at (0.9,-0.8) {};
      \node[ssgvertex, rectangle] at (0.9,-0.5) {};
      \node[ssgvertex,rectangle] (c1) at (2.1,-0.5) {};
      \node[ssgvertex] (c) at (2.1,-0.8) {};
      \node[ssgvertex] (d) at (3,0) {};
      \node[ssgvertex,rectangle]  at (2.7,0) {};
      \node[ssgvertex,rectangle] (e) at (2.1,0.8) {};
      \node[ssgvertex]  at (2.1,1.1) {};
      \node[ssgvertex, rectangle] (f) at (0.9,0.8) {};
      \node[ssgvertex]  at (0.9,1.1) {};

      \draw (a0) edge[->, bend right=20]node[fill=white, inner sep=1pt, pos=0.5] {\tiny $.5$} (b);
      \draw (a0) edge[<-, bend left=20]node[fill=white, inner sep=1pt, pos=0.5] {\tiny $.5$} (b);
      \draw (b) edge[->, thick]node[fill=white, inner sep=1pt, pos=0.4] {\tiny $.5$} (c1);
      \draw (c) edge[->, bend right]node[fill=white, inner sep=1pt, pos=0.5] {\tiny $.5$} (d);
      \draw (d) edge[->, thick, bend right]node[fill=white, inner sep=1pt, pos=0.5] {\tiny $.5$} (e);
      \draw (e) edge[->]node[fill=white, inner sep=1pt, pos=0.5] {\tiny $.5$} (f);
      \draw (f) edge[->, bend right]node[fill=white, inner sep=1pt, pos=0.5] {\tiny $.5$} (a1);
      \draw (a1) edge[->]node[left,  inner sep=1pt, pos=0.4] {\tiny $.5$} (a0);

      \draw (a0) edge[->, dashed]node[fill=white, inner sep=1pt, pos=0.6] {\tiny $.5$} (-0.2, -1.6);
      \draw (b) edge[<-, dashed]node[fill=white, inner sep=1pt, pos=0.5] {\tiny $.5$} (0.8, -1.6);
      \draw (c1) edge[->, dashed]node[fill=white, inner sep=1pt, pos=0.5] {\tiny $.5$} (3, -1.4);
      \draw (c) edge[<-, dashed]node[fill=white, inner sep=1pt, pos=0.5] {\tiny $.5$} (2.4, -1.6);
      \node at (-0.7, -1.35) {\tiny $X_i^-$};
      \node at (-0.7, 1.05) {\tiny $X_i^+$};

      \node[ssgvertex] (Ai) at (-0.7, 0) {};
      \node at (-1, 0) {\tiny $A_i$};
      \draw (Ai) edge[->, dashed]node[fill=white, inner sep=1pt, pos=0.5] {\tiny $.5$} (-0.7, 0.85);
      \draw (Ai) edge[<-, dashed]node[fill=white, inner sep=1pt, pos=0.5] {\tiny $.5$} (-0.7, -1.2);
    \end{scope}
    \begin{scope}[yshift=-9cm,xshift=0cm]
      \node at (1.2, 1.5) {\scriptsize $\gsp'$ and $\yspp$ on $V_i$};
      \node[ssgvertex, fill=black] (a0) at (0,-0.3) {};
      \node[ssgvertex, fill=black, rectangle] (a1) at (0,0.3) {};
      \node[ssgvertex] (b) at (0.9,-0.8) {};
      \node[ssgvertex, rectangle] at (0.9,-0.5) {};
      \node[ssgvertex,rectangle] (c1) at (2.1,-0.5) {};
      \node[ssgvertex] (c) at (2.1,-0.8) {};
      \node[ssgvertex] (d) at (3,0) {};
      \node[ssgvertex,rectangle]  at (2.7,0) {};
      \node[ssgvertex,rectangle] (e) at (2.1,0.8) {};
      \node[ssgvertex]  at (2.1,1.1) {};
      \node[ssgvertex, rectangle] (f) at (0.9,0.8) {};
      \node[ssgvertex]  at (0.9,1.1) {};

      \draw (a0) edge[->, bend right=20]node[fill=white, inner sep=1pt, pos=0.5] {\tiny $0$} (b);
      \draw (a0) edge[<-, bend left=20]node[fill=white, inner sep=1pt, pos=0.5] {\tiny $0$} (b);
      \draw (b) edge[->, thick]node[fill=white, inner sep=1pt, pos=0.4] {\tiny $1$} (c1);
      \draw (c) edge[->, bend right]node[fill=white, inner sep=1pt, pos=0.5] {\tiny $0$} (d);
      \draw (d) edge[->, thick, bend right]node[fill=white, inner sep=1pt, pos=0.5] {\tiny $0$} (e);
      \draw (e) edge[->]node[fill=white, inner sep=1pt, pos=0.5] {\tiny $0$} (f);
      \draw (f) edge[->, bend right]node[fill=white, inner sep=1pt, pos=0.5] {\tiny $0$} (a1);
      \draw (a1) edge[->]node[left,  inner sep=1pt, pos=0.4] {\tiny $0$} (a0);

      \draw (a0) edge[->, dashed]node[fill=white, inner sep=1pt, pos=0.6] {\tiny $0$} (-0.2, -1.6);
      \draw (b) edge[<-, dashed]node[fill=white, inner sep=1pt, pos=0.5] {\tiny $1$} (0.8, -1.6);
      \draw (c1) edge[->, dashed]node[fill=white, inner sep=1pt, pos=0.5] {\tiny $1$} (3, -1.4);
      \draw (c) edge[<-, dashed]node[fill=white, inner sep=1pt, pos=0.5] {\tiny $0$} (2.4, -1.6);
      \node at (-0.7, -1.35) {\tiny $X_i^-$};
      \node at (-0.7, 1.05) {\tiny $X_i^+$};

      \node[ssgvertex] (Ai) at (-0.7, 0) {};
      \node at (-1, 0) {\tiny $A_i$};
      \draw (Ai) edge[->, dashed]node[left, inner sep=2pt, pos=0.5] {\tiny $1$} (-0.7, 0.85);
      \draw (Ai) edge[<-, dashed]node[left, inner sep=2pt, pos=0.5] {\tiny $1$} (-0.7, -1.2);
    \end{scope}
    \begin{scope}[yshift=-9cm,xshift=5.5cm]
      \node at (1.5, 1.5) {\scriptsize $G$ and $F$ on $V_i$};

      \node[ssgvertex, fill=black] (a) at (0,0) {};
      \node[ssgvertex] (b) at (0.9,-0.8) {};
      \node[ssgvertex] (c) at (2.1,-0.8) {};
      \node[ssgvertex] (d) at (3,0) {};
      \node[ssgvertex] (e) at (2.1,0.8) {};
      \node[ssgvertex] (f) at (0.9,0.8) {};
      \draw (b) edge[->, thick] (c);
      \draw (c) edge[->, bend right] (d);
      \draw (d) edge[->, thick, bend right] (e);
      \draw (e) edge[->] (f);
      \draw (f) edge[->, bend right] (a);
      \draw (a) edge[->, dashed]  (-0.7, -1.6);
      \draw (b) edge[<-, dashed]  (0.8, -1.6);
      \draw (c) edge[->, dashed] (3, -1.4);
      \draw (c) edge[<-, dashed] (1.8, -1.6);
      \node at (1.8, -1.8) {\tiny $X_i^-$};
      \node at (-0.7, -1.8) {\tiny $X_i^+$};

    \end{scope}
  \end{tikzpicture}
  \caption{An illustration of the different graphs and flows (restricted to a single component $V_i$) used in our
  algorithm for Local-Connectivity ATSP. The black vertices depict terminals. In the
split graphs, we depict debt vertices by squares and free vertices by circles.}
  \label{fig:overview_lcatsp}
\end{figure}
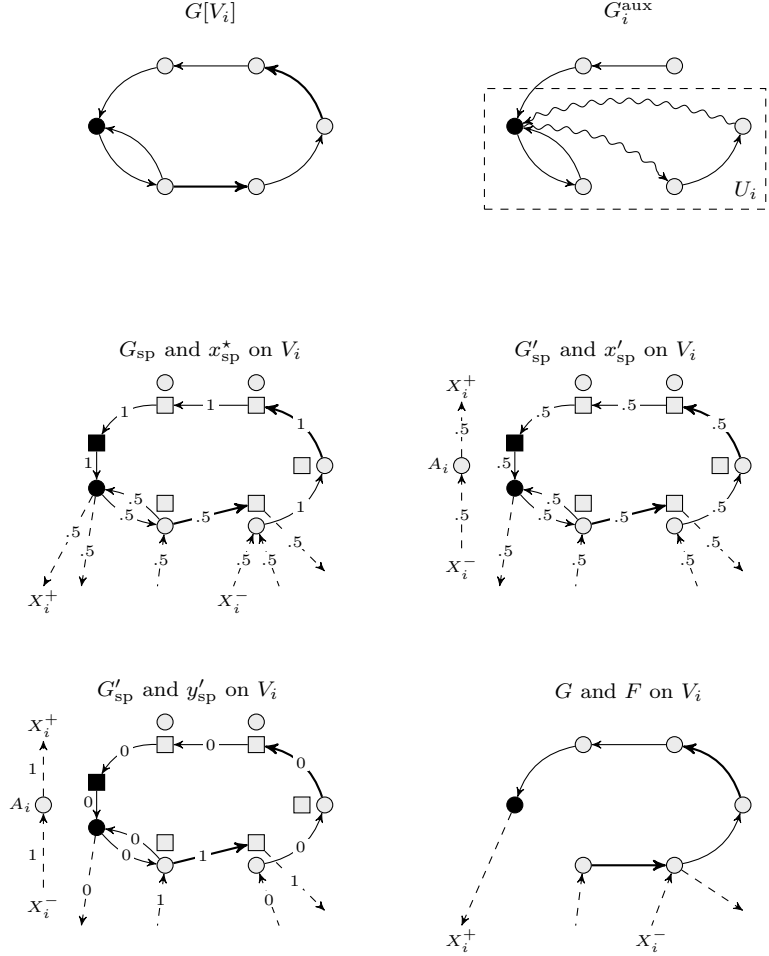
\fi

\paragraph{Transforming $\xspp$ into an integral flow and obtaining our solution $F$.}

In the next step we round $\xspp$ to integrality while respecting degrees of vertices:
\begin{lemma} \label{lem:exists_integral_y}
There exists an integral circulation $\yspp$ on $\gsp'$ satisfying the following conditions:
\iflncs
\else
\begin{itemize}
	\item $\yspp(\delta^-(v)) \le \ceil{2\xsp(\delta^-(v))}$ for each $v \in V(\gsp)$, 
	\item $\yspp(\delta^-(A_i)) = 1$ for each $i$.
\end{itemize}\fi
Such a circulation $\yspp$ can be found in polynomial time.
\end{lemma}
\iflncs\else
\begin{proof}
The bounds are integral, and there exists a fractional circulation which satisfies them, namely $2\xspp$.
\end{proof}\fi

We will now transform $\yspp$ into an Eulerian set of edges $F$ in the original graph $G$.
We can think of this as a three-stage process.

First, we map all edges adjacent to the auxiliary vertices $A_i$ back to their
preimages in $\gsp$, obtaining from $\yspp$ an integral \emph{pseudo-flow} $\ysp$ in
$\gsp$. (We use the term \emph{pseudo-flow} as now, some vertices may not
satisfy flow conservation.)

Second, we contract the two copies $v^0$ and $v^1$ of every vertex $v\in V$,
thus mapping all edges back to their preimages in $G$. (We remove all edges
$(t^1,t^0)$ for $t \in T$.) This creates an integral pseudo-flow $y$ in $G$.

Since the in- and out-degree of $A_i$ were exactly $1$ in $\yspp$, now (in $y$)
in each component $U_i$ there is a pair of vertices $u_i$, $v_i$ which are the
head and tail, respectively, of the mapped-back edges adjacent to $A_i$. These
are the only vertices where flow conservation in $y$ can be
violated.\footnote{It is violated unless $u_i = v_i$.} As the third step, to
repair this, we route a walk $P_i$ from $u_i$ to $v_i$, as described below. Our Eulerian set of
edges $F \subseteq E$ which we finally return is the integral pseudo-flow $y$ plus the
union (over $i$) of all such walks $P_i$, i.e., $\one_{F} = y + \sum_i
\one_{P_i}$.

It remains to describe how we route these paths. Fix $i$. Recall that $U_i$ is strongly connected in $\giaux$. We distinguish two cases:
\begin{itemize}
	\item If $A_i$ is a free vertex or the edge exiting $A_i$ in $\yspp$ (in $\gsp'$) is a debt edge, then select a shortest $u_i$-$v_i$-path in $\giaux$, map each edge of this path to its preimage path (see Definition\nobreakspace \ref {def:auxiliary_graph}) and concatenate them to obtain a $u_i$-$v_i$-walk $P_i$ in $V_i$.\footnote{Note that this walk may exit $U_i$, but it will stay inside $V_i$.}
	\item If $A_i$ is a debt vertex but the edge exiting $A_i$ in $\yspp$ (in $\gsp'$) is a free edge, then by \protect \MakeUppercase {F}act\nobreakspace \ref {fact:backtrack} there is a terminal $t$ inside $U_i$, with a path from $t$ to $v_i$ using only cheap edges.\footnote{Map the path given by \protect \MakeUppercase {F}act\nobreakspace \ref {fact:backtrack} from $\gsp$ to $G$.} Proceed as above to obtain a $u_i$-$t$-walk and then append this cheap $t$-$v_i$-path to it, obtaining a $u_i$-$v_i$-walk $P_i$ in $V_i$.
\end{itemize}

This concludes the description of the algorithm. In \iflncs\else Sections\nobreakspace \ref {sec:light} and\nobreakspace  \ref {sec:cross} \fi we prove that the returned Eulerian set of edges $F$ has the properties we desire, i.e., 

\begin{lemma} \label{lem:light}
For every connected component $\tG$ of $(V,F)$ we have $w(\tG) \le 10 \cdot \lbs(\tG)$.
\end{lemma}

\begin{lemma} \label{lem:cross}
For every component $V_i$ we have $|\delta_{F}^+(V_i)| \ge 1$.
\end{lemma}

Lemmas\nobreakspace \ref {lem:sum_of_lbs_is_low} and\nobreakspace  \ref {lem:light} together prove that our algorithm is $100$-light with respect to $\lb$.

\iflncs\else
\subsubsection{Bounding the cost -- proof of Lemma\nobreakspace \ref {lem:light}} \label{sec:light}

For this section, let us fix $\tG$ to be 
a connected component of $(V,F)$. We want to prove the following lightness
claim: $w(\tG) \le 10 \cdot \lbs(\tG)$.

Intuitively, our solution is cheap because we built the split graph $\gsp$ so
as to ensure that any circulation in $\gsp$ which roughly respects the degree
bounds given by $\xsp$ has low cost, and because our pseudoflow $y$ is not too far
from being a circulation (in particular, the cost of walks $P_i$ can be
accounted for). We make this argument precise below.

First, recall that the edges in $F$ are edges in the integral pseudo-flow
$y$ and the edges of the walks $P_1, P_2, \ldots, P_k$. For a walk $P_i$, we
write $P_i \subseteq \tG$ if $P_i$ is contained in $\tG$ (note that a walk is
either contained in or disjoint from $\tG$ because $\tG$ is a connected component
of $(V,F)$).  Hence
\begin{align}
  w(\tG) = \sum_{e\in E(\tG)} y_e w(e)  + \sum_{i: P_i \subseteq \tG} w(P_i).
  \label{eq:boundcost}
\end{align}

\newcommand{\tGsp}{\tG_{\mathrm{sp}}}
We start by analyzing the first term: $\sum_{e\in E(\tG)} y_e w(e)$. 
Recall that $y$ is obtained by mapping $\ysp$ from $\gsp$ to $G$. Let
$\tGsp$ denote the pre-image of $\tG$ in $\gsp$.
It will be convenient to define the \emph{debt} of $\ysp$ as the number of ``unpaid''  expensive edges in this component:
\begin{align*}
  \debt = \sum_{e\in E(\tGsp): \wsp(e) = w_1} \ysp(e)  - \sum_{t\in V(\tG) \cap T} \ysp(t^1, t^0).
\end{align*}
That is,   the debt of $\ysp$
on $\tGsp$ is the difference between the $\ysp$-flow on all expensive edges inside
$\tGsp$ (those edges ``incur debt'') and the $\ysp$-flow on all edges $(t^1,t^0)$
inside $\tGsp$ (those edges ``discharge debt'').  Note also that, by the definition of $y$ from $\ysp$, $\sum_{e\in E(\tGsp): \wsp(e) = w_1} \ysp(e) = \sum_{e\in E(\tG): w(e) = w_1} y(e)$.

By the construction of $\gsp$ we have the following  upper bound on $\debt$.
\begin{lemma} \label{fact:circulation_has_no_debt}
	Let $e_i^-$ and $e_i^+$ be the incoming and outgoing edges of $A_i$ in $\yspp$. 
	For any $i$, define
	\[
	\bad(i) = \begin{cases}
	1 & \text{if $e_i^-$ is a debt edge and $e_i^+$ is a free edge}, \\
	0 & \text{otherwise.}
	\end{cases}
	\]
	Then
  \begin{align*}
    \debt \leq \sum_{i: P_i \subseteq \tG} \bad(i) 
  \end{align*}
  Moreover, if $\ysp$ is a circulation, then $\debt = 0$. 
\end{lemma}
\begin{proof}
  In $\gsp$, the only edges of the form $(u^0,v^1)$ are the expensive edges,
  and the only edges of the form $(u^1,v^0)$ are edges $(t^1,t^0)$ with $t \in
  T$. Hence, if $\ysp$ is a circulation, then
  \begin{align*}
   \sum_{e \in E(\tGsp) : \wsp(e) = w_1} \ysp(e)  &= \ysp(\delta^-(\{ v^1: v \in V(\tG)\})) \\
  &  = \ysp(\delta^+(\{ v^1: v \in V(\tG) \})) = \sum_{t \in V(\tG) \cap T} \ysp(t^1, t^0), 
\end{align*}
  which shows that $\debt=0$ in this case.

  Now, if $\ysp$ is not a circulation, then one  can  imagine transforming it
  into one by connecting the ``dangling'' endpoints of $e_i^-$ and $e_i^+$ in
  $\gsp$ using a new ``virtual'' edge for each $i$; this edge would be from a debt vertex to a free
  vertex if and only if $e_i^-$ is a debt edge and $e_i^+$ is a free edge.
  Hence, by the calculations above, $\debt$ is upper-bounded by
  the number of such new edges in our component. (We do not have equality as
  the new edges may also be introduced from free vertices to debt vertices.)
\end{proof}

Using the above lemma, degree bounds and basic calculations, we upper-bound the term $\sum_{e\in E(\tG)} y_e w(e)$:
\begin{lemma}
  We have $\sum_{e\in E(\tG)} y_e w(e)  \leq  6\lbs(\tG)$.
  \label{lem:firstsumbound}
\end{lemma}
\begin{proof}
We begin by writing:
\begin{align*}
  \sum_{e\in E(\tG)} y_e w(e) &\leq   w_0 \sum_{e\in E(\tG)} y_e  + w_1 \sum_{e\in E(\tGsp): \wsp(e) = w_1} \ysp(e) \\
  &= w_0 \sum_{e\in E(\tG)} y_e  + w_1 \rb{\debt + \sum_{t\in V(\tG)\cap T} \ysp(t^1, t^0)} \\
  &= w_0 \sum_{v\in V(\tG)} y(\delta^-(v)) + w_1 \sum_{t\in V(\tG) \cap T} \yspp(t^1, t^0) + w_1\cdot \debt.
\end{align*}

We bound the first term. Note that for each $v\in V$ we have
$y(\delta^-(v)) = \ysp(\delta^-(\{v^0, v^{1}\}))$ and $\xs(\delta^-(v))
= \xsp(\delta^-(\{v^0, v^{1}\}))$. Lemma\nobreakspace \ref {lem:exists_integral_y} guarantees that
$\yspp(\delta^-(v^j)) \leq \ceil{2\xsp(\delta^-(v^j))}$ for $j \in \{0,1\}$. This implies that $\yspp(\delta^-(\{v^0, v^1\}))
\leq 2\xs(\delta^-(v)) + 2$,
because:
\begin{itemize}
\item If $v \in T$, then $v^1$ has only one outgoing edge in $\gsp$, which goes to $v^0$, and thus for any circulation $c$ in $\gsp$ we have $c(\delta^-(\{v^0, v^1\})) = c(\delta^-(v^0))$. It follows that
$\yspp(\delta^-(\{v^0, v^1\})) = \yspp(\delta^-(v^0)) \le \ceil{2 \xsp(\delta^-(v^0))} \le 2 \xsp(\delta^-(v^0)) + 1 = 2 \xsp(\delta^-(\{v^0, v^1\})) + 1 = 2 \xs(\delta^-(v)) + 1$.
\item If $v \not \in T$, then there are no edges between $v^0$ and $v^1$ in $\gsp$, and thus for any circulation $c$ in $\gsp$ we have $c(\delta^-(\{v^0, v^1\})) = c(\delta^-(v^0)) + c(\delta^-(v^1))$. It similarly follows that
$\yspp(\delta^-(\{v^0, v^1\})) = \yspp(\delta^-(v^0)) + \yspp(\delta^-(v^1)) \le \ceil{2 \xsp(\delta^-(v^0))} + \ceil{2 \xsp(\delta^-(v^1))} \le 2 \xsp(\delta^-(v^0)) + 2 \xsp(\delta^-(v^1)) + 2 = 2 \xsp(\delta^-(\{v^0, v^1\})) + 2 = 2 \xs(\delta^-(v)) + 2$.
\end{itemize}
 As $\ysp$
is the same as $\yspp$ except for the edges redirected from the auxiliary
vertices in $\gsp$ (which may increase the in-degree of a vertex by at most $1$), 
\begin{align*}
  y(\delta^-(v)) = \ysp(\delta^-(\{v^0, v^{1}\})) \leq \yspp(\delta^-(\{v^0, v^{1}\})) + 1
  \leq 2\xs(\delta^-(v)) + 3 \leq 5 \xs(\delta^-(v)),
\end{align*}
where we used $\xs(\delta^-(v)) \geq 1$ for the last inequality. Therefore
\begin{align*}
  w_0 \sum_{e\in E(\tG)} y_e = w_0 \sum_{v\in V(\tG)} y(\delta^-(v)) \leq 5\cdot w_0\sum_{v\in V(\tG)}\xs(\delta^-(v)). 
\end{align*}
We bound the second term similarly:
\begin{align*}
  \yspp(t^1,t^0) = \yspp(\delta^-(t^1)) \leq  \ceil{2 \xsp(\delta^-(t^1))} \le 2 \ceil{\xsp(\delta^-(t^1))} = 2 \ceil{f(\delta^-(t))}.
\end{align*}
Plugging both in we get:
\begin{align*}
  \sum_{e\in E(\tG)} y_e w(e) &\leq 5\cdot w_0\sum_{v\in V(\tG)}\xs(\delta^-(v)) + 2\cdot w_1 \sum_{t\in (\tG)\cap T} \lceil f(\delta^-(t))\rceil + w_1 \cdot \debt \\
  &\leq 5\cdot \left(w_0\sum_{v\in V(\tG)}\xs(\delta^-(v)) +  w_1 \sum_{t\in V(\tG)\cap T} \lceil f(\delta^-(t))\rceil \right) +w_1  \cdot \debt \\
  & = 5\lbs(\tG) +w_1  \cdot \debt. 
\end{align*}
(For the final equality, recall the definition of the function $\lbs$.) We will be done if we show that
$w_1 \cdot \debt \le \lbs(\tG)$.
By Lemma\nobreakspace \ref {fact:circulation_has_no_debt} we have $\debt\le \sum_{i:P_i
  \subseteq \tG} \bad(i) $. Whenever $\bad(i) = 1$,
\protect \MakeUppercase {F}act\nobreakspace \ref {fact:debt-path} implies that $P_i$ contains a terminal, and so
$\lbs(P_i) \ge w_1$.
Consequently, 
\[
w_1  \cdot \debt\le \sum_{i : P_i \subseteq \tG,
  \bad(i) = 1} w_1 \le \sum_{i : P_i \subseteq \tG,
  \bad(i) = 1} \lbs(P_i) \le \lbs(\tilde G).\]
 In the last inequality we used that the
different $P_i$'s are disjoint subsets of $\tilde G$ (each $P_i$ is
contained in $V_i$). The statement follows.
\end{proof}

It remains to argue that the cost of the walks $P_i$ can be accommodated.

\begin{lemma} \label{lem:path_is_light}
For every $i$ we have $w(P_i) \le 4 \cdot \lbs(P_i)$.
\end{lemma}
\begin{proof}
This holds because in $P_i$, each vertex has low indegree and expensive edges can be offset against terminals. Concretely, we have the following two claims:
\begin{claim}
For each $v \in V(P_i)$ we have $|\delta_{P_i}^-(v)| \le 4$. (In particular, $|E(P_i)| \le 4 |V(P_i)|$.) 
\end{claim}
\begin{proof}
This follows from the fact that we select a \emph{shortest} $u_i$-$v_i$-path (or $u_i$-$t$-path) in $\giaux$ (see Definition\nobreakspace \ref {def:auxiliary_graph}). This implies that each vertex $v$ appears as an internal vertex on at most one preimage path of a prepaid edge (otherwise we could shortcut the path in $\giaux$). Same applies for postpaid edges. And similarly, each $v$ appears as head of at most one preimage path of any kind (single cheap edge, prepaid or postpaid). This means that $v$ has indegree at most $3$ on the walk created from preimages of edges in $\giaux$. The path from $t$ to $u_i$ can contribute a fourth incoming edge.
\end{proof}
\begin{claim}
The number of expensive edges on $P_i$ is at most twice the number of terminals on $P_i$, i.e., $|E_1 \cap E(P_i)| \le 2 \cdot |T \cap V(P_i)|$.
\end{claim}
\begin{proof}
Expensive edges appear only in preimage paths of prepaid or postpaid edges, one per such preimage path, and such a path also contains a terminal (as its head or tail). A terminal can only appear as head of one prepaid and as tail of one postpaid edge (otherwise, again, we could shortcut the path in $\giaux$).
\end{proof}
Having these two claims, we can bound the cost:
\begin{align*}
w(P_i) &\le w_0 \cdot |E(P_i)| + w_1 \cdot |E_1 \cap E(P_i)| \\
       &\le w_0 \cdot 4 |V(P_i)| + w_1 \cdot 2 |T \cap V(P_i)| \\
       &\le 4 \rb{w_0 \cdot |V(P_i)| + w_1 \cdot |T \cap V(P_i)|} \\
       &\le 4 \cdot \lbs(P_i).
\end{align*}
\end{proof}

Now we have all the tools needed to prove our lightness claim, i.e., to upper-bound~\eqref{eq:boundcost} by $10\cdot \lbs(\tG)$:
\begin{align*}
w(\tG) &= \sum_{e\in E(\tG)} y_e w(e)  + \sum_{i: P_i \subseteq \tG} w(P_i) \\
&\le 6 \cdot \lbs(\tG) + \sum_{i : P_i \subseteq \tG} 4 \cdot \lbs(P_i) \\
&\le 6 \cdot \lbs(\tG) + 4 \cdot \lbs(\tG) \\
&= 10 \cdot \lbs(\tG),
\end{align*}
where line 2 is by Lemmas\nobreakspace \ref {lem:firstsumbound} and\nobreakspace  \ref {lem:path_is_light}, and in
line 3 we use that all
walks $P_i$ are disjoint and in $\tG$. This completes the proof of Lemma\nobreakspace \ref {lem:light}.
\qedmanual
\fi

\iflncs\else
\subsubsection{Crossing the cuts -- proof of Lemma\nobreakspace \ref {lem:cross}} \label{sec:cross}

In this section we prove that our solution $F \subseteq E$ crosses each
component $V_i$. Let us recall that $U_i$ is a sink strongly-connected
component of the graph $\giaux$, so there is no edge in $\giaux$ from $U_i$ to
$V_i \setminus U_i$. However, there can be such edges in $G$ and therefore it
might not be the case that the edge in $F$ which leaves $U_i$ also leaves
$V_i$. We will however argue that $F$ does contain some edge leaving $V_i$.
Assume towards a contradiction that this is not true, i.e., that $V_i$ is
a union of connected components of the graph $(V,F)$.

\newcommand{\esp}{e_{\mathrm{sp}}}
Fix $i$. Consider the (only) edge in $\yspp$ (in $\gsp'$) with tail $A_i$; let $\esp$ be its image in $\ysp$ (in $\gsp$) and $e$ its image in $y$ (in $G$). We have $\esp \in X_i^+$; the tail of $e$ is $v_i \in U_i$, and the head of $e$ is in $V_i \setminus U_i$ (since we assumed that no edge of $F$ leaves $V_i$).

The following claim is the first, simplest example of how we use the structure of $\giaux$ to reason about our solution $F$.
\begin{claim} \label{claim:edge_is_expensive}
Any edge from $U_i$ to $V_i \setminus U_i$ in $G$ is expensive. (Therefore $e$ is expensive, and $\esp$ is a debt edge.)
\end{claim}
\begin{proof}
If there was such a cheap edge, it would also appear in $\giaux$. However, there is no edge in $\giaux$ from $U_i$ to $V_i \setminus U_i$.
\end{proof}
Define $(V_i \setminus U_i)^{\mathrm{sp}} = \{ v^0, v^1 : v \in V_i \setminus U_i \} \subseteq V(\gsp)$.
We can now start traversing $\ysp$ inside $(V_i \setminus U_i)^{\mathrm{sp}}$ like an Eulerian graph (since $\ysp$ satisfies flow conservation in $(V_i \setminus U_i)^{\mathrm{sp}}$), starting from $\esp$, until we return to $\Uisp$.
\begin{claim}
We will not reach a terminal before returning to $\Uisp$. (Therefore we will return on a debt edge $\esp'$.)
\end{claim}
\begin{proof}
If there was a terminal $t \in T \cap (V_i \setminus U_i)$ such that there is a path from $\head(\esp)$ (a debt vertex) to $t^1$ in $\gsp[V_i]$, then (assuming without loss of generality that there is no other terminal on this path) all edges of this path are cheap. Therefore $e$ together with the image of this path in $G$ would give rise to a postpaid edge $(\tail(e),t)$ from $U_i$ to $V_i \setminus U_i$ in $\giaux$, a contradiction. Since $\mathrm{head(\esp)}$ is a debt vertex and we do not visit a terminal, we will return on a debt edge.
\end{proof}

Now we distinguish two cases. Suppose that the debt edge $\esp'$ on which we
return to $\Uisp$ is \emph{not} the image in $\gsp$ of the (only) edge in
$\yspp$ (in $\gsp'$) with head $A_i$. This means that we can keep following $y$
inside $\Uisp$ until we exit $\Uisp$ again via some edge $\esp''$ (maybe
$\esp'' = \esp$). By \protect \MakeUppercase {C}laim\nobreakspace \ref {claim:edge_is_expensive}, $\esp''$ is expensive. So
we have followed a path (a segment of $y$) which contains a debt edge and later
an expensive edge; this means (see \protect \MakeUppercase {F}act\nobreakspace \ref {fact:debt-path}) that it must also
contain a terminal in between. Pick the last such terminal $t$; then the
segment of $y$ between $t^0$ and $\tail(\esp'')$ must consist of cheap (free)
edges. This gives rise to a prepaid edge $(t,\head(\esp''))$ from $U_i$ to $V_i
\setminus U_i$ in $\giaux$, a contradiction.

So suppose instead that the debt edge $\esp'$ on which we return to $\Uisp$ is,
in fact, the image in $\gsp$ of the (only) edge in $\yspp$ (in $\gsp'$) with
head $A_i$. This means that $\esp' \in X_i^-$. Recall that by definition, $X_i^-$ consists only of debt or only of free edges, and therefore
every edge in $X_i^-$ must be a debt edge.
By \protect \MakeUppercase {F}act\nobreakspace \ref {fact:backtrack}, $\gsp[U_i]$ contains a path from a terminal $t^0$ via
cheap edges to the expensive edge $\esp$. Again, the image of this path in $G$
gives rise to a prepaid edge $(t,\head(\esp))$ from $U_i$ to $V_i \setminus
U_i$ in $\giaux$, a contradiction.  This concludes the proof of
Lemma\nobreakspace \ref {lem:cross}. \qedmanual
\fi

\iflncs
\else \ifsvjour
	\bibliographystyle{plain}
	\bibliography{atsp}
\else
	\section*{Acknowledgment}
	An earlier version of this paper has appeared in the proceedings of IPCO 2016 (the 18th Conference on Integer Programming and Combinatorial Optimization).
	This work has been published in Mathematical Programming Series B. The final publication is available at Springer via \href{http://dx.doi.org/10.1007/s10107-017-1195-7}{http://dx.doi.org/10.1007/s10107-017-1195-7}.

	\bibliographystyle{alphaurl}
	\bibliography{atsp}
\fi \fi

\end{document}